\documentclass[11pt]{article}
\usepackage[margin=1in]{geometry}
\usepackage[T1]{fontenc}
\usepackage[utf8]{inputenc}
\usepackage{mathpazo}
\usepackage{amsmath,amsthm,mathtools,amssymb}
\usepackage[shortlabels]{enumitem}
\usepackage[svgnames,table]{xcolor}
\usepackage{graphicx}
\usepackage{booktabs}
\usepackage{multirow}
\usepackage{caption}
\usepackage{algorithm}
\usepackage{algorithmic}

\usepackage{natbib}
\usepackage{bm}
\usepackage{pifont}
\definecolor{DarkBlue}{rgb}{0.1,0.1,0.5}
\definecolor{DarkGreen}{rgb}{0.1,0.5,0.1}
\usepackage{hyperref}
\hypersetup{colorlinks=true, linkcolor=DarkBlue, urlcolor=DarkBlue, citecolor=DarkGreen}
\usepackage[capitalize]{cleveref}

\usepackage{booktabs}

\newtheorem{lemma}{Lemma}
\newtheorem{theorem}{Theorem}

\newtheorem{corollary}{Corollary}
\newtheorem{definition}{Definition}
\newtheorem{observation}{Observation}

\newtheorem{remark}{Remark}
\newtheorem{proposition}{Proposition}

\crefname{OP}{open problem}{open problems}
\Crefname{OP}{Open Problem}{Open Problems}

\newcommand{\cmark}{\ding{51}}   % check mark
\newcommand{\xmark}{\ding{55}}   % cross mark
\newcommand{\qmark}{\textbf{?}}
\newcommand{\rfp}[1]{{\scriptsize [Prop.~\ref{#1}]}}
\newcommand{\rft}[1]{{\scriptsize [Thm.~\ref{#1}]}}
\newcommand{\rfc}[1]{{\scriptsize [Cor.~\ref{#1}]}}
\newcommand{\EFX}{\textsc{EFX}}
\newcommand{\EF}{\textsc{EF}}
\newcommand{\EQ}{\textsc{EQ}}
\newcommand{\EQX}{\textsc{EQX}}
\newcommand{\EQone}{\textsc{EQ1}}
\newcommand{\EFone}{\textsc{EF1}}
\newcommand{\EFonew}{\makebox[2.4em][c]{\EF{1}}}
\newcommand{\EFXw}{\makebox[2.4em][c]{\EFX}}
\newcommand{\EQonew}{\makebox[2.4em][c]{\EQ{1}}}
\newcommand{\EQXw}{\makebox[2.4em][c]{\EQX}}
\newcommand{\MMS}{MMS}
\newcommand{\MNW}{MNW}
\newcommand{\PO}{\textsc{PO}}
\newcommand{\NPH}{\textsc{NP-hard}}

\title{Simultaneous Envy and Equitability Guarantees}

\author{
    \parbox[t]{0.45\textwidth}{%
    \centering
    Hadi Hosseini \\ 
    Penn State University\\ 
    \small \texttt{hadi@psu.edu}
    }
    \and
    \parbox[t]{0.45\textwidth}{%
    \centering
    Shraddha Pathak\\
    Penn State University\\ \small \texttt{ssp5547@psu.edu}
    } 
    \vspace{1em}
    \and
    \parbox[t]{0.45\textwidth}{%
    \centering
    Lirong Xia\\
    Rutgers University\\ \small \texttt{lirong.xia@rutgers.edu}
    }
    \and
    \parbox[t]{0.45\textwidth}{%
    \centering
    Chengkai Zhang\\
    Rutgers University\\ \small \texttt{cz521@scarletmail.rutgers.edu}
    }
}

\date{}

\begin{document}
\maketitle

\begin{abstract}
Recent work in fair division has focused on either simultaneously satisfying closely related fairness notions or achieving a single notion across the ex-ante and ex-post worlds. 
We study the compatibility of two fundamentally different fairness notions: envy-freeness and equitability. 
For indivisible goods-only and chores-only settings, we study the existence and complexity of simultaneously satisfying their relaxations, revealing sharp contrasts between the two settings. For normalized binary goods, we give a polynomial-time algorithm for computing an \EFone{}+\EQone{} allocation with at most seven agents, but also construct a normalized instance with a larger number of agents for which no such allocation exists.
In sharp contrast, binary chores admit the stronger \EFX{}+\EQX{} guarantee for any number of agents, even without normalization.
We further initiate the study of cross-notion ex-ante--ex-post guarantees, asking whether randomized allocations can provide ex-ante guarantees for one notion while preserving ex-post guarantees for another.
\end{abstract}

\section{Introduction} \label{sec:intro}

Fair division of indivisible items is a fundamental problem concerning the allocation of resources or tasks among a group of agents with potentially different idiosyncratic preferences over those items.
The literature in this field has given rise to a rich landscape of fairness notions, motivated by both normative principles of distributive justice \citep{rawls1971theory} and practical considerations such as computational and existential limitations (see, e.g., \citet{amanatidis2023fair}).

Most existing work studies fairness notions in isolation. However, simultaneously satisfying multiple notions can better accommodate the diversity of metrics (e.g. in human preferences as discussed in \citet{herreiner2009envy,hosseini2024fairness}) while providing stronger fairness guarantees.
The simultaneous guarantees have often emerged as a byproduct of  axiomatic implications or algorithmic procedures, with a few recent notable exceptions \citep{akrami2025achieving,akrami2026simultaneous,akrami2026achieving,babaioff2021fair}. 
A separate line of research explores best-of-both-worlds guarantees, aiming to achieve a \textit{single} fairness notion (e.g., envy-freeness) ex ante while preserving an approximate guarantee for the same notion ex post \citep{aziz2024best,bhaskar2025best,babaioff2021best}.
This raises a natural question: 

\begin{quote}
\textit{Can we simultaneously guarantee multiple, fundamentally different notions of fairness, both in deterministic allocations and across ex-ante and ex-post worlds?}
\end{quote}

We consider two prominent fairness notions. The first, \textit{envy-freeness} (\EF{}) \citep{foley1966resource}, is based on \textit{intra}personal comparisons: each agent evaluates its own bundle against others' bundles using its \textit{own valuation function}. The second, \textit{equitability} (EQ) \citep{dubins1961cut}, relies on \textit{inter}personal comparisons, requiring all agents to derive equal \textit{subjective} value from their respective bundles.
This distinction leads to different invariance properties. Envy-freeness is individually \textit{scale-invariant}: rescaling an agent's valuation preserves envy-freeness. Equitability, however, can be affected by scaling up or down a single individual's valuation. Thus, we consider a mild \textit{normalization}: all agents agree on the value of the grand bundle---a well-motivated assumption in divisible cake cutting \citep{brams1996fair}, online allocations \citep{gkatzelis2021fair}, and in allocating mixtures of goods and chores \citep{barman2026fair,hosseini2026landscape}.

In the presence of indivisible items, exact \EF{} and \EQ{} allocations may fail to exist ex-post (e.g., one item and two agents).
Thus, we focus on their natural relaxations: \textit{envy-freeness up to one item} (\EFone{}) and \textit{envy-freeness up to any item} (\EFX{}), as well as their equitability counterparts \EQone{} and \EQX{}. These notions require the corresponding fairness guarantee to hold after the hypothetical removal of one item from an appropriate bundle.

Our work opens a new direction for investigating the interplay between fairness notions. 
We study the compatibility of two envy-based and equitability-based notions in a \textit{deterministic} sense under structured valuations. Importantly, we initiate, to the best of our knowledge, the study of achieving distinct fairness guarantees across the ex-ante and ex-post worlds.
In particular, we ask whether probabilistic guarantees for one notion (e.g., \EF{}) can coexist with deterministic guarantees for another (e.g., \EQone{}).

\subsection{Contributions}

We investigate the compatibility of relaxed notions of envy-freeness and equitability in deterministic allocations of indivisible items, as well as their relationships with probabilistic ex-ante guarantees. We focus on \textit{additive} valuations in two settings: when all items are goods (with non-negative values) and when all items are chores (with non-positive values). 
\Cref{tab:results} summarizes our main results.

\paragraph{Deterministic guarantees.}

We show that an allocation satisfying \EFone{} + \EQone{} may fail to exist for instances with $n\ge 3$ agents, even for normalized bivalued valuations (\Cref{prop:non-exist-3-and-more}). Furthermore, deciding whether such an allocation exists is \NPH{} for unnormalized instances with two agents  (\Cref{thm:hardness}).
Nonetheless, for two agents, \EFone{}+\EQone{} allocations always exist and admit polynomial-time algorithms for both goods (\Cref{thm:two_agents_ef1_eq1}) and chores (\Cref{cor:two_agents_chores}) under normalization. We further show that the stronger \EFX{}+\EQX{} guarantee exhibits a sharp distinction between the two settings: it is achievable for goods, but may not exist for chores (\Cref{prop:two_agents_efx_eqx}).

Our main results in this section concern normalized binary goods.
We design an algorithm achieving \EFone{}+\EQone{} for instances with up to seven agents (\Cref{thm:binary-7}). 
This guarantee does not hold for all normalized binary goods instances: we construct an instance (with 113 agents and 341 goods) for which no complete \EFone{}+\EQone{} allocation exists (\Cref{thm:binary-113-counterexample}). Despite the large instance, the proof is human-verifiable and intuitive, but it is unclear how to extend it to instances with few agents.
Thus, the problem for instances with between eight and 112 agents remains unresolved. 
Under laminar approval structures, however, we show that \EFone{}+\EQone{} allocations exist and can be computed in polynomial time for any number of agents (\Cref{thm:laminar_binary_goods}). 
In contrast, for binary chores, we show that the stronger \EFX{}+\EQX{} guarantee can always be satisfied, even for instances that are not normalized (\Cref{thm:binary_chores_efx_eqx}).

\paragraph{Ex-ante and ex-post guarantees.}

For two agents with normalized valuations, both goods and chores admit randomized allocations achieving ex-ante \EF{} + \EQ{}, with every allocation in the support satisfying \EFone{} ex-post (\Cref{prop:ex-ante-EQQ-ex-post-EF1-2-agents}).
For normalized binary instances, we further observe a sharp distinction between goods and chores. While we construct randomized allocations achieving ex-ante \EF{} + \EQ{} in both settings, the ex-post guarantees differ: we obtain \EQone{} for goods (\Cref{thm:binary-lottery}), but the stronger \EFone{}+\EQone{} guarantee for chores, which holds even \textit{without normalization} and extends to restricted-additive chores (\Cref{thm:chores-binary-ex-ante-ex-post}); the latter guarantee is due to \citet{sun2025randomized}, and we include a short self-contained proof of the form we need.
Our counterexample in the deterministic setting for normalized binary goods also rules out randomized allocations with an ex-post \EFone{}+\EQone{} guarantee, regardless of the ex-ante requirements (\Cref{cor:binary-113-lottery}). Finally, for laminar binary goods the ex-post \EFone{} guarantee cannot be strengthened to \EFX{} (\Cref{prop:laminar-expost-efx-tight}), and for restricted-additive chores neither \EFone{} nor \EQone{} can be strengthened to its ``up to any item'' counterpart (\Cref{prop:chores-expost-tight}).

\begin{table*}[t]
\centering
\small
\caption{Summary of our main results: 
\cmark{} = an allocation always exists; \xmark{} = it may fail to exist; P = it exists and can be computed in polynomial time; $\qmark$ = it is an open problem; and (general) refers to cases where results hold \textit{without normalization}. Unless stated otherwise, all instances are \textit{additive} and \textit{normalized}.  
A randomized guarantee is a randomized allocation that satisfies the stated ex-ante notions exactly, while every allocation in its support satisfies the stated ex-post guarantee.
\Cref{thm:chores-binary-ex-ante-ex-post} and \Cref{cor:rest-add-chores-ef1-eq1} follow from \citet{sun2025randomized}; all other resolved entries are established in this paper.}
\label{tab:results}
\setlength{\tabcolsep}{3pt}
\begin{tabular}{clllllll}
\toprule
\multirow{2}{*}{Domains} & \multirow{2}{*}{Valuation Types} & \multirow{2}{*}{Agents}
& \multicolumn{2}{l}{Deterministic Guarantees} & \multicolumn{3}{l}{Randomized Guarantees}\\
\cmidrule(lr){4-5}\cmidrule(lr){6-8}
& & & Guarantee & Ref. & Ex-ante & Ex-post & Ref.\\
\midrule

\multirow{6}{*}{\textbf{Goods}}
& \multirow{2}{*}{Additive} & \multirow{2}{*}{$2$}
  & \EFXw+\EQXw~\cmark & \rfp{prop:two_agents_efx_eqx}
  & \EQ{} & \EFXw~\xmark & \rfp{prop:no-ex-ante-EQ-ex-post-EFX}\\
&&& \EFonew+\EQonew~P & \rft{thm:two_agents_ef1_eq1}
  & \EF{}+\EQ{} & \EFonew~P & \rft{prop:ex-ante-EQQ-ex-post-EF1-2-agents}\\

\cmidrule(lr){2-8}

& \multirow{3}{*}{Binary} & $n\le7$
  & \EFonew+\EQonew~P & \rft{thm:binary-7}
  &  &  & \\
&& $8\le n\le112$
  & \EFonew+\EQonew~\qmark & \scriptsize [Open]
  &  &  & \\
&& Arbitrary
  & \EFonew+\EQonew~\xmark & \rft{thm:binary-113-counterexample}
  & \EF{}+\EQ{} & \EQonew~P & \rft{thm:binary-lottery}\\

\cmidrule(lr){2-8}

& Binary \& laminar & Arbitrary
  & \EFonew+\EQonew~P & \rft{thm:laminar_binary_goods}
  & \EF{}+\EQ{} & \EFonew+\EQonew~P & \rft{thm:laminar-binary-lottery}\\

\midrule

\multirow{4}{*}{\textbf{Chores}}
& \multirow{2}{*}{Additive} & \multirow{2}{*}{$2$}
  & \EFXw+\EQXw~\xmark & \rfp{prop:two_agents_efx_eqx}
  & \multirow{2}{*}{\EF{}+\EQ{}} & \multirow{2}{*}{\EFonew~P} & \multirow{2}{*}{\rft{prop:ex-ante-EQQ-ex-post-EF1-2-agents}}\\
&&& \EFonew+\EQonew~P & \rfc{cor:two_agents_chores} &&&\\

\cmidrule(lr){2-8}

& Binary (general) & \multirow{2}{*}{Arbitrary}
  & \EFXw+\EQXw~P & \rft{thm:binary_chores_efx_eqx}
  & \multirow{2}{*}{\EF{}+\EQ{}} & \multirow{2}{*}{\EFonew+\EQonew~P} & \multirow{2}{*}{\rft{thm:chores-binary-ex-ante-ex-post}}\\
& Rest.-add. (general) && \EFonew+\EQonew~P & \rfc{cor:rest-add-chores-ef1-eq1} &&&\\

\bottomrule
\end{tabular}
\end{table*}

\subsection{Related Work}\label{sec:related-work}

Fair division of indivisible items has developed into a rich research area at the intersection of economics, artificial intelligence, and theoretical computer science; we refer to \citet{brandt2016handbook} and \citet{amanatidis2023fair} for broad overviews.
Here we discuss the lines of work most relevant to this paper.

\paragraph{Envy-freeness and its relaxations.}
Envy-free allocations of indivisible items may fail to exist even in trivial instances, which has motivated a hierarchy of relaxations.
\citet{lipton2004approximately} introduced the envy-cycle-elimination procedure, which guarantees, for arbitrary monotone valuations, an allocation in which envy is bounded by a single item; the resulting notion was later formalized as \EFone{} by \citet{budish2011combinatorial} and has become the standard benchmark; \EFone{} allocations always exist and can be computed in polynomial time.
The stronger notion \EFX{} \citep{caragiannis2019unreasonable} is known to exist for two agents and for identical (general monotone) valuations via the leximin++ solution \citep{plaut2020almost}---a technique we reuse for our two-agent \EFX{}+\EQX{} guarantee---for three additive agents \citep{chaudhury2020efx}, and for dichotomous valuations, i.e., valuations with binary marginals \citep{babaioff2021fair,bu2023efx}.
Very recently, its general existence was refuted for monotone submodular valuations \citep{akrami2026counterexample,mackenzie2026counterexamples}; the additive case remains a central open problem in the field.
Our results concern the opposite frontier: even where \EFone{}---the weakest notion in this hierarchy---is trivially achievable, we show it can be incompatible with equitability.

\paragraph{Equitability and its relaxations.}
The systematic study of \EQone{} and \EQX{} for indivisible goods was initiated by \citet{freeman2019equitable}, who showed that leximin allocations are \EQX{} (and Pareto optimal when all values are strictly positive), gave a pseudopolynomial-time algorithm for \EQone{}+\PO{}, and delineated the complexity of combining equitability with welfare objectives.
The companion paper on chores \citep{freeman2020equitable} established \EQone{}+\PO{} in pseudopolynomial time and showed that, in contrast to goods, leximin may violate \EQone{}; \citet{sun2023equitability} completed the complexity landscape of equitability together with welfare maximization for both goods and chores.
Beyond additive valuations, \citet{hosseini2026landscape} showed that \EQone{} allocations exist and are efficiently computable for two agents with general valuations and for many agents under doubly monotone valuations, while existence may fail---and is intractable to decide---in general, and \citet{barman2026fair} established equitability within additive margins for non-negative valuations.
Finally, two recent works connect equitability to envy: \citet{wang2025achieving} characterize when \EQ{} and \EF{} can be achieved \emph{simultaneously} when monetary subsidies are allowed, and \citet{wang2026centralized} show that \EFone{} is always compatible with an equitability notion imposed across \emph{groups} of agents.
Our work complements these by studying the compatibility of individual-level relaxations (\EQone{}/\EQX{} with \EFone{}/\EFX{}) without transfers.

\paragraph{Simultaneous fairness guarantees.}
A recurring theme in fair division is whether several desiderata can be achieved by a single allocation.
The most prominent combinations pair fairness with \emph{efficiency}: maximum Nash welfare allocations are \EFone{}+\PO{} \citep{caragiannis2019unreasonable}, equitable analogues were given by \citet{freeman2019equitable,freeman2020equitable}, and the long-standing existence question of \EFone{}+\PO{} for chores was recently resolved positively \citep{mahara2025existence}.
Closer to our agenda are combinations of two \emph{fairness} notions.
\citet{garg2025exploring} map out the logical implications among more than twenty fairness notions; their question---whether every allocation satisfying one notion satisfies another---is complementary to ours, which asks whether \emph{some} allocation satisfies both.
\citet{akrami2025achieving} and \citet{akrami2026simultaneous} combined \MMS{} approximations with \EFX{}/\EFone{}, and \citet{akrami2026achieving} achieved \EFone{} and epistemic-\EFX{} simultaneously.
To the best of our knowledge, the envy-equitability axis has so far been addressed only with monetary transfers \citep{wang2025achieving}, at the group level \citep{wang2026centralized}, or within a single notion across the stages of a randomized allocation \citep{bhaskar2025best}; our paper provides the first systematic study of deterministic and randomized compatibility between item-level relaxations of envy-freeness and equitability, for both goods and chores.

\paragraph{Structured valuation classes.}
Restricted preference domains often admit substantially stronger guarantees.
For binary (dichotomous) marginals, truthful mechanisms yield Lorenz-dominating allocations that are simultaneously \EFX{} and \MNW{} \citep{babaioff2021fair}; for matroid-rank valuations, \MNW{} allocations are \EFone{} and utilitarian-optimal \citep{benabbou2020finding}, and a range of justice criteria can be computed efficiently, also under entitlements \citep{viswanathan2023general,suksompong2023weighted,bu2023efx}.
For bivalued goods, \MNW{} implies \EFX{} \citep{amanatidis2021maximum}, and for \emph{personalized} bivalued goods, \citet{jin2025pareto} characterized Pareto optimality and proved that \EFX{} allocations always exist.
These results sharpen the message of our impossibility construction: on normalized personalized-bivalued instances with strictly positive values, \EFX{} allocations exist \citep{jin2025pareto} and \EQX{} allocations exist via leximin \citep{freeman2019equitable}, yet we show that no allocation may satisfy even \EFone{} and \EQone{} \emph{together}.
Our positive results for normalized binary goods (up to seven agents, and any number of agents under laminar approval sets) identify conditions that restore compatibility. Binary structure alone is insufficient: \Cref{thm:binary-113-counterexample} establishes nonexistence for a normalized binary instance with 113 agents.

\paragraph{Fair division of chores.}
Fairness notions for chores are not mirror images of their goods counterparts, and the chores landscape is generally harder: \EFX{} allocations need not always exist for additive chores \citep{he2026efx}, whereas exact \EFX{} (with \PO{}) is achievable in polynomial time for binary chores \citep{tao2023existence}.
On the equitability side, \citet{freeman2020equitable} and \citet{sun2023equitability} give existence and complexity results for \EQone{} chores.
Our binary-chores algorithm complements \citet{tao2023existence}: instead of pairing \EFX{} with efficiency, we pair it with \EQX{}, again without normalization and for any number of agents; for the broader class of restricted-additive chores we obtain \EFone{}+\EQone{}.

\paragraph{Randomized allocations and best-of-both-worlds fairness.}
The ``best of both worlds'' (BoBW) paradigm asks for randomized allocations that are exactly fair in expectation while every realized allocation retains an approximate guarantee: \citet{aziz2024best} showed that ex-ante \EF{} together with ex-post \EFone{} is always achievable in polynomial time (and that adding efficiency leads to impossibilities).
For binary valuations, \citet{halpern2020fair} show that lexicographically tie-broken \MNW{} is group-strategyproof, \EFone{}, and Pareto optimal, and that fractional \MNW{} can be implemented as a lottery over deterministic \MNW{} allocations.
Subsequent work obtained share-based BoBW guarantees \citep{babaioff2021best}, extensions to subadditive valuations \citep{feldman2023breaking}, and strengthened ex-post guarantees toward \EFX{}: in particular, \citet{bu2024best} achieve ex-ante \EF{} with ex-post \EFX{} for two agents, and with ex-post \EFX{}+fractional-\PO{} for bivalued goods.
Closest to our randomized results is the recent BoBW treatment of \emph{equitability} by \citet{bhaskar2025best}: randomized allocations that are ex-ante \EQ{} and ex-post \EQone{} always exist for two agents, but may fail to exist for three or more agents, where existence is strongly \NPH{} to decide; positive results reappear under binary valuations.
For chores, \citet{sun2025randomized} give a mechanism, \textsc{RandChore}, that is group strategyproof in expectation and simultaneously ex-ante \EF{}, \EQ{} and \textsc{PROP} and ex-post \EFone{}, \EQone{} and \textsc{PROP1}, as well as ex-ante and ex-post \PO{}, whenever costs are \emph{$1$-restricted additive}---exactly the restricted-additive chores of \Cref{sec:prob_chores}.
With the exception of \citet{sun2025randomized}, all of the above relax a single fairness family across the two stages (or combine envy with shares).
Our randomized allocations are instead \emph{cross-notion}: they are simultaneously ex-ante \EF{} \emph{and} ex-ante \EQ{}, while every support allocation satisfies item-level relaxations of the two notions---ex-post \EFone{} for normalized two-agent instances (via the two-agent partitions of \citealp{kyropoulou2020almost}, which are \EFone{} under either assignment of bundles), ex-post \EQone{} for normalized binary goods (strengthened to ex-post \EFone{}+\EQone{} under laminar approval sets), and---by the result of \citet{sun2025randomized} recalled in \Cref{sec:prob_chores}---ex-post \EFone{}+\EQone{} for restricted-additive chores.
We further show that the ex-post components of the last two guarantees are best possible (\Cref{sec:expost-limits}).

 \section{Preliminaries}

\paragraph{Problem instance.}
For any $k \in \mathbb{N}$, let $[k] \coloneqq \{1,2,\dots,k\}$.
An \textit{instance} consists of a set of $n$ agents, $N = [n]$, a set of $m$ items, $M = \{o_1,o_2,\dots,o_m\}$, and a valuation profile $\{v_1,v_2,\dots,v_n\}$. 
Each agent $i \in N$ has an \textit{additive valuation} function $v_i : 2^{M} \to \mathbb{R}$ satisfying $v_i(\emptyset)=0$ and $v_i(S) = \sum_{o \in S} v_i(o)$ for every $S \subseteq M$, where we write $v_i(o)$ instead of $v_i(\{o\})$ for simplicity.
An instance is a \textit{goods-instance} if for every $i\in N$ and $o\in M$ $v_i(o)\ge 0$, and it is a \textit{chores-instance} if $v_i(o)\le 0$ for every $i\in N$ and $o\in M$.
We sometimes use $g$ or $c$ instead of $o$ to refer to a good or a chore, respectively.

\paragraph{Valuations.}
A valuation profile is \textit{normalized} if $v_i(M)=v_j(M)$ for every $i,j\in N$.
It is \textit{identical} if $v_i=v$ for every $i\in N$. 
It is \textit{binary} if for every $i\in N$ and $o\in M$, $v_i(o)\in \{0,1\}$ in case of goods, and $v_i(o)\in \{0,-1\}$ for chores. 
A valuation profile is \textit{bivalued} if, for every agent $i$, there exist $p_i,q_i\in \mathbb{R}$ such that $v_i(o)\in \{p_i,q_i\}$ for every item $o\in M$.
If every $o\in M$ has a base value $v(o)$ and $v_i(o)\in \{0,v(o)\}$ for every agent $i\in N$, then the instance is \emph{restricted additive}.
Both bivalued and restricted additive instances are generalizations of binary valuations.

\paragraph{Allocation.}
An allocation $A=(A_1,\dots,A_n)$ is an $n$-partition of the set of items $M$, i.e., $\bigcup_{i\in N} A_i = M$ and $A_i \cap A_j = \emptyset$ for every distinct $i,j\in N$.
Each bundle $A_i$ is allocated to agent $i$, and we refer to $v_i(A_i)$ as agent $i$'s \emph{(realized) utility}; in chores instances, we call $-v_i(A_i)$ her \emph{disutility} (or cost).
A \emph{randomized allocation} $\bm{A}$ (a.k.a. \emph{lottery}) is a probability distribution over deterministic allocations $A^{(1)},\dots,A^{(k)}$.
Let $p_\ell\in [0,1]$ denote the probability of $A^{(\ell)}$ for every $\ell \in [k]$ and $\sum_{\ell\in [k]} p_\ell = 1.$

\paragraph{Deterministic Fairness.} 
An allocation $A$ is \emph{equitable} (\EQ{}) if $v_i(A_i)=v_j(A_j)$ for every $i,j\in N$.
In a goods instance, an allocation $A$ is \emph{equitable up to one good \EQone{}} if for every $i,j\in N$ such that $A_j\neq\emptyset$ we have $v_i(A_i)\ge v_j(A_j\setminus \{g\})$ for some $g\in A_j$.
It is \emph{equitable up to any good} (\EQX{}) if the inequality holds for every $g\in A_j$.
For chores, $A$ is \EQone{} (respectively, \EQX{}) if, for every $i,j\in N$ with $A_i\neq\emptyset$, the inequality $ v_i(A_i\setminus\{c\})\ge v_j(A_j)$ holds for some (respectively, every) $c\in A_i$.

An allocation $A$ is \emph{envy-free} (\EF{}) if $v_i(A_i)\ge v_i(A_j)$ for all $i,j\in N$.
In a goods instance, $A$ is \emph{envy-free up to one good \EFone{}} if for every $i,j\in N$ such that $A_j\neq\emptyset$, we have $v_i(A_i)\ge v_i(A_j\setminus \{g\})$ for some good $g\in A_j$.
It is \emph{envy-free up to any good} (\EFX{}) if the inequality holds for every $g\in A_j$. For chores, $A$ is \EFone{} (respectively, \EFX{}) if, for every $i,j\in N$ with $A_i\neq\emptyset$, the inequality $ v_i(A_i\setminus\{c\})\ge v_i(A_j)$ holds for some (respectively, every) $c\in A_i$.
Note that under identical valuations the corresponding envy and equitability notions coincide.

\paragraph{Randomized (ex-ante) Fairness.}
A randomized allocation $\bm{A}$ is \emph{ex-ante equitable} if all agents have the same expected utility; i.e., for every $i,j\in N$, 
$
\sum_{\ell\in [k]} p_\ell v_i(A_i^{(\ell)})
=
\sum_{\ell\in [k]} p_\ell v_j(A_j^{(\ell)}).
$
It is \emph{ex-ante envy-free} if no agent envies another in expectation; i.e., for every $i,j\in N$,
$
\sum_{\ell\in [k]} p_\ell v_i(A_i^{(\ell)})
\ge
\sum_{\ell\in [k]} p_\ell v_i(A_j^{(\ell)}).
$
$\bm{A}$ satisfies a fairness notion \emph{ex-post} if every allocation in its support satisfies that notion; e.g., $\bm{A}$ is ex-post \EFone{} if every $A^{(\ell)}$ in the support of $\bm{A}$ is \EFone{}.

\section{Deterministic Guarantees}
\label{sec:deterministic}

\subsection{Fundamental Barriers}

To warm up, we show that \EFone{} and \EQone{} remain incompatible even under normalized and bivalued valuations, and deciding whether an instance admits such a solution is \NPH{} even for two agents.

\begin{proposition}\label{prop:non-exist-3-and-more}
    For $n\ge 3$ agents with normalized bivalued valuations, an \EFone{}+\EQone{} allocation may not exist.
\end{proposition}

\begin{proof}[Proof Sketch]
The instance consists of $m=2n$ goods and is depicted below: agents $1,\ldots,n-1$ value $g_1$ at $1$ and every other good at $\varepsilon$, while agent $n$ values every good at $\delta:=\frac{1+(2n-1)\varepsilon}{2n}$, where $\varepsilon>0$ is small enough that $(2n-1)\varepsilon<\delta$.
\[
\begin{array}{c|cccc}
 & g_1 & g_2 & \cdots & g_{2n}\\
\hline
\text{agents } 1,\ldots,n-1 & 1 & \varepsilon & \cdots & \varepsilon\\
\text{agent } n & \delta & \delta & \cdots & \delta
\end{array}
\]
Since some agent $i<n$ misses $g_1$ and thus realizes utility below $\delta$, \EQone{} forces agent $n$ to receive at most one good; by pigeonhole, some agent then holds at least three goods, whom agent $n$ \EFone{}-envies.
The full proof and an analogous construction for chores appear in \Cref{sec:appendix_challenges}.
Notably, all but one agent have identical valuations: although the envy-based and equitability-based notions coincide under identical valuations, a single deviating agent already destroys simultaneous existence.
\end{proof}

\begin{remark}\label{rem:common-bivalued}
Personalization is essential to the particular construction in \Cref{prop:non-exist-3-and-more}: under a \emph{common} bivalued domain, where $v_i(o)\in\{p,q\}$ with the same $p>q\ge 0$ for all agents, normalization forces every agent to value the same number of items at $p$, so no agent can play the role of the flat agent $n$ above. It is not essential to non-existence itself: \Cref{thm:binary-113-counterexample} gives a counterexample for normalized binary valuations ($p=1,q=0$) with 113 agents. 
\end{remark}

Beyond non-existence, deciding whether a given (possibly unnormalized) instance admits an \EFone{}+\EQone{} allocation is intractable.

\begin{theorem}\label{thm:hardness}
For two agents with additive (possibly unnormalized) valuations, deciding whether an allocation that is both \EFone{} and \EQone{} exists is \NPH{}.
\end{theorem}

The proof is deferred to \Cref{sec:appendix_hardness}; it establishes a general case for the up-to-$k$-item versions of both notions.

The normalization assumption is necessary for the simultaneous satisfaction of \EFone{} and \EQone{}. Consider two agents and four goods, where one agent values every good at 1 and the other values every good at 0.
No \EFone{} solution that allocates all the items is \EQone{}.
Furthermore, even under normalized binary goods valuations, neither of the two guarantees in \EFone{}+\EQone{} can, in general, be strengthened to its respective ``up to any'' counterpart (see \Cref{sprop:six-agent-both}).
This motivates the study of restricted settings, such as two-agent instances and binary valuations (\Cref{sec:two_agents} and \Cref{sec:binary_deterministic}), and probabilistic compatibility across ex-ante and ex-post worlds (\Cref{sec:exante-expost}).

Towards positive results, we first consider normalized two-agent instances with general additive valuations. 
We then turn to binary valuations, where simultaneous guarantees can be obtained for a larger number of agents. 

\subsection{Two-Agent Instances} \label{sec:two_agents}

Notice that the impossibility in \Cref{prop:non-exist-3-and-more} applies to instances with at least three agents. 
For two-agent goods instances, normalization is sufficient to guarantee the simultaneous existence of the stronger notions \EFX{} and \EQX{}.
However, \EFX{}+\EQX{} allocations need not exist for chores instances, even for two agents.

\begin{proposition} \label{prop:two_agents_efx_eqx}
For normalized goods instances with two agents, an \EFX{}+\EQX{} allocation always exists. However, such an allocation may not exist for normalized chores instances with two agents. 
\end{proposition}

The existence guarantee follows via the leximin++ allocation for the instance that maximizes the value of the worst-off agent, and subject to that maximizes the number of items in the bundle of this agent. 
The complete proof as well as the counterexample for chores are presented in \Cref{sec:appendix_two_agent_deterministic}.

Although \EFX{}+\EQX{} may fail for chores, both goods and chores admit the weaker combination \EFone{}+\EQone{} in polynomial time.
For goods, the algorithm maintains bundles $A_1,A_2$ and realized utilities $U_i=v_i(A_i)$. 
At each iteration, an agent $i$ with lower realized utility receives a remaining good $g$ that maximizes $v_i(g)-v_{j}(g)$.

\begin{algorithm}[tb]
\caption{\textsc{Greedy-Balance} for two-agent goods}
\label{alg:two_agent_goods}
\textbf{Input}: A normalized two-agent goods instance\\
\textbf{Output}: An allocation $(A_1,A_2)$
\begin{algorithmic}[1]
\STATE Initialize $A_1,A_2\leftarrow\emptyset$,
$U_1,U_2\leftarrow 0$, and $R\leftarrow M$.
\WHILE{$R\neq\emptyset$}
    \IF{$U_1\leq U_2$}
        \STATE Choose
        $g\in\arg\max_{h\in R}\{v_1(h)-v_2(h)\}$.
        \STATE $A_1\leftarrow A_1\cup\{g\}$ and
        $U_1\leftarrow U_1+v_1(g)$.
    \ELSE
        \STATE Choose
        $g\in\arg\max_{h\in R}\{v_2(h)-v_1(h)\}$.
        \STATE $A_2\leftarrow A_2\cup\{g\}$ and
        $U_2\leftarrow U_2+v_2(g)$.
    \ENDIF
    \STATE $R\leftarrow R\setminus\{g\}$.
\ENDWHILE
\STATE \textbf{return} $(A_1,A_2)$.
\end{algorithmic}
\end{algorithm}

The following invariant connects the equitability and envy comparisons.

\begin{lemma}\label{lem:cross_bundle_dominance}
At every stage of \Cref{alg:two_agent_goods}, $v_1(A_1)\geq v_2(A_1) \text{ and } v_2(A_2)\geq v_1(A_2).$
\end{lemma}

\begin{proof}
We prove the first inequality; the second is symmetric. 
Define $\Delta(g)=v_1(g)-v_2(g)$. 
Suppose that, at some stage, $\Delta(A_1)<0$. 
Since the valuations are normalized, $\Delta(M)=0$, and hence there exist $g\in A_1$ and $h\notin A_1$ such that $\Delta(g)<0<\Delta(h).$
Consider the iteration in which $g$ was assigned to agent~$1$. 
The good $h$ must already have been allocated; otherwise, the algorithm would have selected $h$ instead of $g$. 
Thus, $h$ was assigned earlier to agent~2. 
At that earlier iteration, however, $g$ was still available and satisfied $\Delta(g)<\Delta(h)$. 
Since agent~$2$ chooses a remaining good minimizing $\Delta(\cdot)$, the algorithm should have selected $g$ rather than $h$, a contradiction.
\end{proof}

\begin{theorem}
\label{thm:two_agents_ef1_eq1}
For normalized goods instances with two agents, \Cref{alg:two_agent_goods} computes an \EFone{}+\EQone{} allocation in polynomial time.
\end{theorem}

\begin{proof}
Let $A=(A_1,A_2)$ be the returned allocation. 
Without loss of generality, assume $v_1(A_1)\leq v_2(A_2).$
If $A_2=\emptyset$, then \Cref{lem:cross_bundle_dominance} implies that both agents value both bundles at zero, and the claim follows.
Let $g$ be the last good assigned to agent~$2$. 
Immediately before receiving $g$, agent~$2$ had realized utility no greater than that of agent~$1$.
Moreover, every good allocated after $g$ was assigned to agent~$1$.
By monotonicity, $ v_1(A_1)\geq v_2(A_2\setminus\{g\}),$ establishing \EQone{}.

It remains to prove \EFone{}. 
Assume agent~$1$ envies agent~$2$. 
Since $A_2\setminus\{g\}$ was agent~2's bundle at an intermediate stage of the algorithm, we apply \Cref{lem:cross_bundle_dominance} at that stage. 
Combining with the \EQone{} guarantee above gives
\[v_1(A_1) \geq v_2(A_2\setminus\{g\}) \geq v_1(A_2\setminus\{g\}),\]
which is precisely the \EFone{} condition. 
Agent~2 does not envy agent~1, since $v_2(A_2)\ge v_1(A_1)\ge v_2(A_1)$ by \Cref{lem:cross_bundle_dominance}.
The algorithm performs $m$ iterations and can be implemented in polynomial time.
\end{proof}

The same balancing idea, with the inequalities and selection rule adjusted for disutilities, yields the analogous result for chores; see \Cref{sec:appendix_two_agent_deterministic} for a complete proof.

\begin{corollary}
\label{cor:two_agents_chores}
For normalized chores instances with two agents, \Cref{alg:utility-balance-chores} computes an \EFone{}+\EQone{} allocation in polynomial time.
\end{corollary}

\subsection{Binary Valuations}
\label{sec:binary_deterministic}

The incompatibility result of \Cref{prop:non-exist-3-and-more} motivates studying the more structured setting of binary valuations. Binary valuations are both practically relevant and theoretically well studied, forming the basis of many recent advances in fair division \citep{halpern2020fair,barman2018greedy,babaioff2021fair,benabbou2020finding,bu2023efx}. 
For normalized binary goods, we establish a polynomial-time guarantee for up to seven agents and a non-existence result for 113 agents. 
In contrast, we give the stronger \EFX{}+\EQX{} guarantee for binary chores for any number of agents and without normalization.

\subsubsection{Goods}

We focus on simultaneously achieving
\EFone{}+\EQone{}, and design a polynomial-time algorithm when there are at most seven agents.
We then show that such a guarantee cannot extend to all normalized binary goods instances: \Cref{thm:binary-113-counterexample} presents a 113-agent counterexample; the gap between instances with more than seven agents but fewer than 113 remains open.

\paragraph{Algorithm description.}
The algorithm proceeds in three phases: a flow computation first constructs an egalitarian-optimal partial allocation that raises every agent to the largest utility level $t$ attainable by all agents simultaneously; utility-preserving swaps then bring the residual goods into a structured form; and, finally, each remaining good is “scattered” to an agent who does not value it.

A detailed description: Our algorithm starts by removing all goods valued by no agent and storing them in $J$.  Let $t$ be the largest integer for which every agent can simultaneously receive $t$ valued goods.  Among all partial allocations $P=(P_1,\ldots,P_n)$ satisfying
\[
    t\le |P_i|\le t+1 \qquad\text{for every }i\in N,
\]
where every good in $P_i$ is valued by $i$, choose one maximizing the number of assigned goods (equivalently, the number of agents receiving $t+1$ goods).
Both $t$ and $P$ can be found by integral network flow.

Let $R=M\setminus\bigcup_{i\in N}P_i$.  In the residual alternating graph, direct an edge from an agent $i$ to each good in $P_i$, and direct an edge from an unowned good or a good owned by another agent to each agent who values it.
Let $S$ be the set of agents reachable from $R$.  For
$
    Z:=R\cup\bigcup_{i\in S}P_i
$
and each $g\in Z$, define the type of $g$ by
\[
    \tau(g):=\{i\in S:v_i(g)=1\}.
\]
As shown below, $\tau(g)$ is a nonempty subset of $S$.
We call $g$ proper-type if $\tau(g)\subsetneq S$, and full-type if $\tau(g)=S$.

\begin{algorithm}[tp]
\caption{\textsc{Level-and-Scatter} for at most seven agents}
\label{alg:level-scatter}
\begin{algorithmic}[1]
\small
\REQUIRE A normalized binary goods instance with $n\le 7$ agents
\ENSURE An allocation $(A_1,\ldots,A_n)$

\STATE Remove from $M$ all goods that every agent values at $0$ and store them in $J$.
\STATE Compute the maximum feasible egalitarian level $t$ and set $B\gets t+1$.
\STATE Compute a maximum partial allocation $P=(P_1,\ldots,P_n)$ such that every good in $P_i$ is valued by $i$ and $t\le |P_i|\le B$ for every $i$.
\STATE Let $R\gets M\setminus\bigcup_{i\in N}P_i$.
\STATE Construct the residual alternating graph and let $S$ be the set of agents reachable from $R$.

\WHILE{there exist $g\in R$, $h\in S$, and $g'\in P_h$ such that $\tau(g)=S$ and $\tau(g')\neq S$}
    \STATE $P_h\gets (P_h\setminus\{g'\})\cup\{g\}$.
    \STATE $R\gets (R\setminus\{g\})\cup\{g'\}$.
\ENDWHILE

\STATE Initialize $A_i\gets P_i$ for every $i\in N$.
\FOR{each $g\in R$}
    \STATE Set $T\gets\tau(g)$.
    \STATE Choose $j\in N\setminus T$ such that
    $v_i(A_j\cup\{g\})\le B$ for every $i\in T$.
    \STATE $A_j\gets A_j\cup\{g\}$.
\ENDFOR
\STATE Distribute the goods in $J$ arbitrarily.
\RETURN $(A_1,\ldots,A_n)$.
\end{algorithmic}
\end{algorithm}

The difficulty is that a residual good cannot be assigned to an agent who values it without increasing that agent's realized utility from $t+1$ to $t+2$, thereby violating \EQone{}. We therefore assign each residual good $g$ to an agent outside its type $\tau(g)$, who values it at zero. Such an assignment must also ensure that no agent in $\tau(g)$ values the recipient's resulting bundle above $t+1$, as otherwise \EFone{} may fail. Lemmas~\ref{lem:closure}-\ref{lem:blocking-seven} reduce the absence of a feasible recipient to a counting obstruction. Lemma~\ref{lem:small-population-seven} rules out this obstruction for $n\leq 7$, separately for residual goods of proper and full type.

We next establish the structural and counting properties needed for this argument.

\begin{lemma}
\label{lem:closure}
Let $P$, $R$, and $S$ be as defined immediately before the swap phase.  Then:
\begin{enumerate}
    \item every agent in $S$ receives exactly $B=t+1$ goods in $P$;
    \item every good in $Z=R\cup\bigcup_{i\in S}P_i$ is valued only by agents in $S$; and
    \item every good valued by an agent outside $S$ is assigned to an agent outside $S$.  Consequently, for $k:= v_i(M)$, we have
    \[
        k\le \sum_{j\notin S}|P_j|.
    \]
\end{enumerate}
\end{lemma}

\begin{proof}
If some $i\in S$ had $|P_i|=t$, an alternating path from a residual good to $i$ could be augmented.
All internal agents on the path would keep the same number of valued goods, while $i$ would gain one.
The resulting partial allocation would still give every agent between $t$ and $B$ valued goods and would assign one additional good, contradicting the maximality of $P$.
This proves the first claim.

The goods reachable in the residual alternating graph are exactly $R\cup\bigcup_{i\in S}P_i$.
From every reachable good there is an outgoing edge to each agent who values it (other than its owner, who is already in $S$ when the good is assigned inside $S$); hence every such valuer is reachable and belongs to $S$.
This proves the second claim.
For the third, a good valued by an agent outside $S$ can therefore belong neither to $R$ nor to a bundle $P_i$ with $i\in S$.
All $k$ goods valued by that agent are consequently contained in the bundles outside $S$, whose total cardinality is $\sum_{j\notin S}|P_j|$.
\end{proof}

\begin{lemma}
\label{lem:outsideS}
If $R\neq\emptyset$, then $|N\setminus S|\ge 2$.
\end{lemma}

\begin{proof}
Maximality of the egalitarian level $t$ implies that some agent receives only $t$ goods in $P$.
By Lemma~\ref{lem:closure}(1), this agent lies outside $S$, so $N\setminus S$ is nonempty.
If $N\setminus S=\{j\}$, then $|P_j|=t$ and Lemma~\ref{lem:closure}(3) gives $k\le t$.
On the other hand, $R\neq\emptyset$ implies $S\neq\emptyset$, and every agent in $S$ receives $t+1$ valued goods, so $k\ge t+1$, a contradiction.
\end{proof}

Assume henceforth that $R\neq\emptyset$, and write
\[
    s:=|S|,\qquad r:=|N\setminus S|,\qquad B:=t+1.
\]
By Lemma~\ref{lem:outsideS}, $r\ge2$.  Moreover, at least one agent outside $S$ receives $B-1$ goods, and every other outside agent receives at most $B$.
Lemma~\ref{lem:closure}(3) therefore yields the mass bounds
\begin{equation}
    k\le \sum_{j\notin S}|P_j|\le rB-1,
    \label{eq:mass-bound-seven}
\end{equation}
and, for every $i\in S$,
\begin{equation}
    k-B\le (r-1)B-1.
    \label{eq:outside-row-bound-seven}
\end{equation}

\begin{lemma}
\label{lem:swap-invariants-seven}
Each iteration of the swap phase preserves every agent's utility, the maximality of $P$, the set $Z$, and the property that every good in $Z$ is valued only by agents in $S$.  The swap phase terminates, and at termination either no good $g\in R$ has $\tau(g)=S$, or every good in $\bigcup_{i\in S}P_i$ has type $S$.
\end{lemma}

\begin{proof}
The agent $h$ values both exchanged goods: $g$ is valued by every agent in $S$, while $g'\in P_h$ is valued by its recipient.
Thus $h$'s utility and bundle size remain $B$, and all other bundles are unchanged.
Exactly one assigned and one residual good are exchanged, so the number of assigned goods, and hence the maximality of $P$, are preserved.
The exchange takes place entirely within $Z$, so $Z$ and its closure property are unchanged.
Finally, the number of full-type goods in $R$ decreases by one in every swap, because $g$ has type $S$ and $g'$ does not; hence the phase terminates after at most $|M|$ swaps.
If a full-type residual good remains at termination, the absence of a further swap implies that no bundle $P_h$ with $h\in S$ contains a good whose type is a proper subset of $S$.
\end{proof}

\begin{lemma}
\label{lem:blocking-seven}
Immediately before any iteration of the scattering phase,
\begin{equation}
    v_i(A_j)\le B\qquad\text{for all }i,j\in N.
    \label{eq:scattering-invariant-seven}
\end{equation}
Let $g$ be the residual good processed in that iteration, let $T=\tau(g)$, and write $q=|T|$. If no feasible recipient exists, then some
agent $i\in T$ values at least
\begin{equation}
    B\left\lceil\frac{n-q}{q}\right\rceil
    \label{eq:blocking-count-seven}
\end{equation}
goods contained in bundles whose recipients lie outside $T$.
\end{lemma}

\begin{proof}
The invariant \eqref{eq:scattering-invariant-seven} holds initially because $|P_j|\le B$ for every $j$.
When a good of type $T$ is assigned, its value is zero to agents outside $T$, while the choice of its recipient preserves the bound for agents in $T$.
Thus the invariant is maintained inductively.

There are $n-q$ candidate recipients in $N\setminus T$.
If a candidate $j$ cannot receive $g$, then, by the invariant, some $i\in T$ must satisfy $v_i(A_j)=B$; otherwise adding $g$ would preserve the cap for every valuer.
Charge each blocked bundle to one such agent.
Some $i\in T$ is charged at least $\lceil(n-q)/q\rceil$ bundles.
These disjoint bundles each contain $B$ goods valued by $i$, which proves the claim.
\end{proof}

The following elementary inequalities are the only point at which the bound of seven agents is used.

\begin{lemma}
\label{lem:small-population-seven}
Let $n=s+r\le7$ with $r\ge2$.  Then
\begin{align}
    \left\lceil\frac{n-q}{q}\right\rceil&\ge r-1
    &&\text{for every }1\le q\le s-1,
    \label{eq:proper-type-inequality-seven}\\
    \left\lceil\frac{r}{s}\right\rceil&\ge r-s.
    \label{eq:full-type-inequality-seven}
\end{align}
\end{lemma}

\begin{proof}
For $1\le q\le s-1$,
\[
    q(r-1)\le(s-1)(r-1)=sr-n+1.
\]
Since $s+r=n$ and $n\le7$,
\[
    sr\le\left\lfloor\frac{n^2}{4}\right\rfloor\le2n-2.
\]
Therefore $q(r-1)\le n-1<n$, or equivalently $(n-q)/q>r-2$, which proves
\eqref{eq:proper-type-inequality-seven}.

For \eqref{eq:full-type-inequality-seven}, if $r\le s+1$, then $r-s\le1\le\lceil r/s\rceil$.
If $r\ge s+2$, the inequality $s+r\le7$ implies $s\le2$.
For $s=1$ the claim is immediate.
For $s=2$, we have $r\le5$ and $\lceil r/2\rceil\ge r-2=r-s$.
\end{proof}

The two possible kinds of residual goods---proper type and full type---are handled separately next.

\begin{lemma}
\label{lem:proper-type-seven}
If a residual good $g$ satisfies $\tau(g)=T\subsetneq S$, then the scattering phase has a feasible recipient for $g$.
\end{lemma}

\begin{proof}
Let $q=|T|\le s-1$ and suppose that $g$ has no feasible recipient.
By Lemmas~\ref{lem:blocking-seven} and~\ref{lem:small-population-seven},
some $i\in T$ values at least $(r-1)B$ goods in bundles of agents outside
$T$, and therefore outside her original bundle $P_i$.
However, $P_i$ contains $B$ of the $k$ goods valued by $i$, so
\eqref{eq:outside-row-bound-seven} implies that at most $(r-1)B-1$ goods
valued by $i$ lie outside $P_i$, a contradiction.
\end{proof}

\begin{lemma}
\label{lem:full-type-seven}
If a residual good $g$ satisfies $\tau(g)=S$ after the swap phase, then the scattering phase has a feasible recipient for $g$.
\end{lemma}

\begin{proof}
By Lemma~\ref{lem:swap-invariants-seven}, every good in the $sB$-good core $\bigcup_{h\in S}P_h$ is valued by every agent in $S$.
Thus, for each $i\in S$, the number of $i$-valued goods outside this core is at most
\begin{equation}
    k-sB\le(r-s)B-1,
    \label{eq:full-type-mass-supp}
\end{equation}
where the inequality follows from \eqref{eq:mass-bound-seven}.
Because the residual full-type good itself lies outside the core and is valued by $i$, we also have $k-sB\ge1$.
Thus, if $r\le s$, the displayed upper bound is already a contradiction and no full-type residual good can remain; it remains to consider $r\ge s+1$.

A full-type good can be assigned only to one of the $r$ agents outside $S$.
If every one of these bundles were blocked,
Lemma~\ref{lem:blocking-seven} would give some $i\in S$ that values at least
\[
    B\left\lceil\frac{r}{s}\right\rceil\ge(r-s)B
\]
goods in outside bundles, where the last inequality follows from Lemma~\ref{lem:small-population-seven}.
This contradicts \eqref{eq:full-type-mass-supp}.
Hence a feasible recipient exists.
\end{proof}

The preceding lemmas cover all configurations.  
Indeed, if $R\neq\emptyset$, then $S\neq\emptyset$ and Lemma~\ref{lem:outsideS} gives $r\ge2$.  
Every residual type is nonempty and is either a proper subset of $S$ or equal to $S$; Lemmas~\ref{lem:proper-type-seven} and \ref{lem:full-type-seven} handle these two exhaustive cases.   
For $n\le2$, these conditions are incompatible, so $R$ must be empty.  
Lemma~\ref{lem:small-population-seven} applies to every feasible pair $(s,r)$.

\begin{theorem}
\label{thm:binary-7}
For normalized binary goods instances with at most seven agents, \Cref{alg:level-scatter} computes an \EFone{} + \EQone{} allocation in polynomial time.
\end{theorem}

\begin{proof}
If $R=\emptyset$, no scattering is needed.  Every realized utility is $t$ or $B=t+1$, and each $P_j$ has at most $B$ goods.  Adding junk goods does not affect any value. Hence the allocation is \EQone{}, and $v_i(A_j)=v_i(P_j)\le B\le v_i(P_i)+1=v_i(A_i)+1$ for all $i,j$, which implies \EFone{}.

Now suppose $R\neq\emptyset$.  Lemma~\ref{lem:swap-invariants-seven} shows that the swap phase terminates without changing any utility or closure property.  
During the scattering phase, every residual good has either proper or full type, so Lemmas~\ref{lem:proper-type-seven} and \ref{lem:full-type-seven} guarantee a feasible recipient at every iteration, regardless of the processing order.  
The returned allocation is therefore complete.

Every scattered good is assigned to an agent outside its type and is thus worth zero to its recipient.  The junk goods are worth zero to everyone.
Consequently, all realized utilities remain in $\{t,t+1\}$, which implies \EQone{}.  Moreover, the invariant \eqref{eq:scattering-invariant-seven} continues to hold after all residual and junk goods are assigned.  Therefore, for all agents $i,j$,
\[
    v_i(A_j)\le t+1\le v_i(A_i)+1.
\]
If $i$ envies $j$, then $v_i(A_j)-v_i(A_i)=1$ and $A_j$ contains a good valued by $i$; removing that good eliminates the envy.
Thus the allocation is \EFone{}.

Finally, $t$ and the maximum partial allocation $P$ can be computed by polynomial-time integral flow.  The alternating reachable set is found by graph search, the swap phase performs at most $|M|$ swaps, and the scattering phase performs at most $|M|$ assignments with polynomially many value checks.
Hence the algorithm runs in polynomial time.
\end{proof}

The seven-agent bound enters the analysis only through
\Cref{lem:small-population-seven}; Remark~\ref{srem:beyond-seven} in \Cref{sec:appendix_7_agent_barriers} explains how the counting argument breaks for $n\ge 8$. Its eight-agent example still admits an \EFone{}+\EQone{} allocation and obstructs only a particular scattering step. We now give a different construction in which no complete \EFone{}+\EQone{} allocation exists.

\paragraph{Non-existence with 113 agents.}
\label{sec:binary-counterexample}
For binary goods, let $u_i=v_i(A_i)$ denote agent $i$'s realized utility. The definitions of \EFone{} and \EQone{} thus are equivalent to
\begin{equation}
\label{eq:binary-fairness-cap}
    v_i(A_j)\le u_i+1
    \quad\text{and}\quad
    u_j\le u_i+1,
    \qquad\text{for all }i,j\in N.
\end{equation}
%Indeed, removing one good reduces the relevant utility by at most one; conversely, whenever a comparison is violated before removal, a good valued at one can be removed to eliminate a gap of at most one. No efficiency or non-wastefulness restriction is imposed: goods may be assigned to agents who value them at zero.

\begin{theorem}[A normalized binary counterexample]
\label{thm:binary-113-counterexample}
There is a normalized binary goods instance with $n=113$ agents and $m=341$ goods, in which every agent approves exactly $k=165$ goods, but no complete \EFone{}+\EQone{} allocation exists.
\end{theorem}

\begin{proof}
Let $P$ be the set of all non-zero binary vectors of length five. Thus, $P=\{0,1\}^5\setminus\{(0,0,0,0,0)\},$ i.e., $|P|=2^5-1=31.$
For every $x\in P$, we create $11$ distinct goods $g_{x,1},\ldots,g_{x,11}$, so the total number of goods is $11\cdot 31 = 341$.

We now define $31$ valuation types, one for each $a\in P$. Write $a=(a_1,\ldots,a_5)$ and $x=(x_1,\ldots,x_5)$. An agent of type $a$ values the good $g_{x,\ell}$ at $1$ exactly when the dot product $ a_1x_1+\cdots+a_5x_5$ is even. All other goods are valued at $0$ by this agent. 

For any fixed non-zero vector $a$, let $r$ be a coordinate for which $a_r=1$. Pair each binary vector $x$ with the vector obtained by flipping only its $r$-th coordinate. Within each pair, the values of $a_1x_1+\cdots+a_5x_5$ have opposite parity. Thus, exactly one vector in each pair gives an even sum. Since there are $32$ vectors, exactly $16$ give an even sum, one of which is the zero vector, which is not in
$P$. Thus, each valuation type approves $15\cdot 11=165$ goods.

Fix one type $a^\star\in P$ and create $83$ agents of type $a^\star$ (all agents here share the same valuations over all goods). We also create one agent of each of the other $30$ types. Thus, the instance has $83+30=113$ agents and $31\cdot 11=341$ goods. Since every agent approves exactly $165$ goods, the instance is normalized.

Towards a contradiction, let us assume that this constructed instance admits a complete allocation $A=(A_1,\ldots,A_{113})$ that is \EFone{}+\EQone{}. 
All $83$ agents of type $a^\star$ approve the same set of $165$ goods, and since their bundles are disjoint, the sum of their utilities is at most $165$. Since $165 < 2\cdot 83$, at least one of these agents has utility at most one; let $i$ be such an agent.

In a binary-goods instance, \EQone{} implies that every agent has utility at most $u_i+1$ since if some agent had utility greater than $u_i+1$, removing any single good would leave that agent with utility greater than $u_i$. Since $u_i\le 1$, every agent therefore has utility at most two, i.e., $u_j\le 2$  for every $j.$

We next show that every four distinct goods have a common agent that values each of them at $1$.
Let their labels be $x_1,x_2,x_3,x_4\in P$. To show that these goods have a common agent that values all of them, we need a non-zero binary vector $a=(a_1,\ldots,a_5)$ such that $a_1x_{r,1}+\cdots+a_5x_{r,5}$ is even for each $r\in\{1,2,3,4\}$. 
These are four parity equations in five binary variables.
After solving the equations, at least one variable remains free. The all-zero vector is a solution, and choosing the free variable to be one gives a non-zero solution as well.
The agent corresponding to this non-zero vector values all four goods at $1$.

Now consider any bundle $A_j$. If it contained at least four goods, choose four of them. By the previous argument, some agent $i$ would approve all four. Removing any one good from $A_j$ would then leave that agent valuing $A_j$ at least three. But agent $i$'s own utility is at most two, so \EFone{} would be violated. Hence every bundle contains at most three goods.
This is impossible because the allocation is complete and $3\cdot 113 = 339 < 341$. 
Therefore, no complete \EFone{}+\EQone{} allocation exists.
\end{proof}

The proof combines a utility bound obtained from identical agents with a bundle-size bound obtained from overlapping approval sets; we present a lemma in \Cref{sec:counter-example-generalization} that isolates this construction.

This example gives an upper bound of 113 on the smallest number of agents for which a normalized binary goods instance can fail to admit an \EFone{}+\EQone{} allocation. However, the example does not assert non-existence for every $n\ge 113$.
Narrowing the gap by resolving $n\in [8,112]$ is an important future direction, either via counterexamples for smaller $n$ or existence and algorithmic guarantees for $n\ge 8$. 

For instances with arbitrarily large $n$, we restrict the class of binary valuations to laminar approval sets and show that under this restriction, \EFone{}+\EQone{} allocations always exist.

\paragraph{Laminar approval sets.}

The bound on the number of agents can be removed when the agents' approval sets have additional structure. For each agent $i$, let
\[
\Gamma_i:=\{g\in M:v_i(g)=1\}
\]
denote her approval set. A family of approval sets is \emph{laminar} if any two sets are disjoint or one contains the other.

\begin{theorem}
\label{thm:laminar_binary_goods}
For normalized binary goods instances with laminar approval sets, \EFone{}+\EQone{} allocations always exist and can be computed in polynomial time, even when the instance has an arbitrary number of agents.
\end{theorem}

\begin{proof}
Normalization and laminarity imply that any two agents have either identical or disjoint approval sets. Thus, the agents can be partitioned into groups $N_1,\ldots,N_q$ of identical agents whose approval sets $\Gamma^{(1)},\ldots,\Gamma^{(q)}$ are pairwise disjoint. Since the instance is normalized, every approval set has the same size, say $k$.

Let $r=\max_{\ell}|N_\ell|$ and $u=\lfloor k/r\rfloor$. Within each group, assign every agent either $u$ or $u+1$ goods from the group's approval set whenever possible. Any remaining goods are distributed among the agents of a fixed group of size $r$, balancing the number of goods from each approval set separately and taking already assigned goods from that set into account.
Every agent consequently receives utility $u$ or $u+1$, yielding \EQone{}. Moreover, no bundle contains more than $u+1$ goods from any single approval set. Hence every agent values every bundle by at most $u+1$ while receiving utility at least $u$, which establishes \EFone{}.
\end{proof}

The binary assumption in \Cref{thm:laminar_binary_goods} is essential: there exists a normalized restricted-additive goods instance with three agents and laminar approval sets that admits no \EFone{}+\EQone{} allocation.

Agent $i$'s \emph{approval set} is $\Gamma_i=\{o\in M:v_i(o)>0\}$, and a family of approval sets is \emph{laminar} if any two of its members are either disjoint or nested.
The proof of \Cref{thm:laminar_binary_goods} uses the following structural observation.

\begin{observation}
\label{sobs:laminar-types}
In a normalized, restricted-additive instance with laminar approval sets, any two agents $i,j$ either have disjoint approval sets or satisfy $v_i=v_j$.
\end{observation}

\begin{proof}
If the sets intersect, laminarity lets us assume without loss of generality that $\Gamma_i\subseteq\Gamma_j$.
Restricted additivity gives $v_i(M)=\sum_{o\in\Gamma_i}v(o)$ and $v_j(M)=\sum_{o\in\Gamma_j}v(o)$, with every base value $v(o)>0$ for $o\in\Gamma_j$.
If $\Gamma_i\subsetneq\Gamma_j$, then $v_i(M)<v_j(M)$, contradicting normalization.
Thus $\Gamma_i=\Gamma_j$, and both agents value each $o\in\Gamma_i$ at $v(o)$ and everything else at $0$, i.e., $v_i=v_j$.
\end{proof}

Consequently, such an instance partitions the agents into groups of identical agents with pairwise disjoint approval sets, which is exactly the structure used in \Cref{thm:laminar_binary_goods}.
The binary assumption in \Cref{thm:laminar_binary_goods} cannot be dropped:

\begin{proposition}
\label{sprop:laminar-ra}
For three agents with normalized restricted-additive valuations and laminar approval sets, an allocation of goods that is both \EFone{} and \EQone{} may fail to exist.
\end{proposition}

\begin{proof}
Consider $m=7$ goods: one \emph{heavy} good $h$ with base value $v(h)=6$, and six \emph{light} goods $\ell_1,\dots,\ell_6$ with $v(\ell_k)=1$.
The approval sets are
\[
\Gamma_1=\Gamma_2=\{h\},
\qquad
\Gamma_3=\{\ell_1,\dots,\ell_6\},
\]
each agent valuing the goods in her approval set at their base values and all other goods at $0$:
\begin{center}
\begin{tabular}{c|ccccccc}
 & $h$ & $\ell_1$ & $\ell_2$ & $\ell_3$ & $\ell_4$ & $\ell_5$ & $\ell_6$\\
\hline
$a_1$ & $6$ & $0$ & $0$ & $0$ & $0$ & $0$ & $0$\\
$a_2$ & $6$ & $0$ & $0$ & $0$ & $0$ & $0$ & $0$\\
$a_3$ & $0$ & $1$ & $1$ & $1$ & $1$ & $1$ & $1$
\end{tabular}
\end{center}
Every agent values the grand bundle at $6$ (the instance is normalized), the valuations have the form $v_i(o)\in\{0,v(o)\}$ (restricted additive), and the two distinct approval sets $\{h\}$ and $\{\ell_1,\dots,\ell_6\}$ are disjoint, hence laminar.

In any allocation, at least one of agents $a_1,a_2$ realizes utility $0$, since their common approval set contains the single good $h$.
To maintain \EQone{} against this zero-utility agent, agent $a_3$ can receive at most one good from her approval set $\Gamma_3=\{\ell_1,\ldots,\ell_6\}$.
Hence at least five light goods go to agents $a_1$ and $a_2$, so one of them receives at least three light goods.
Agent $a_3$'s value for that bundle is at least $3$, and at least $2$ after the removal of any single good, while her own utility is at most $1$; hence she \EFone{}-envies that agent.
Therefore no allocation is simultaneously \EFone{} and \EQone{}.
\end{proof}

\subsubsection{Chores}
The situation is considerably more favorable for chores. 
In particular, neither normalization nor a bound on the number of agents is required, and both fairness guarantees can be strengthened from their ``up to one'' versions to their ``up to any'' versions.

\begin{theorem}
\label{thm:binary_chores_efx_eqx}
For binary chores instances, an \EFX{}+\EQX{} allocation always exists and can be computed in polynomial time, even when the instance has an arbitrary number of agents and unnormalized valuations.
\end{theorem}

\paragraph{Notation.}
For binary chores we use the following notation.
The set of chores that have value $0$ for agent $i$ is denoted by $M_i^0:=\{c\in M:v_i(c)=0\}$.
We write $M^0:=\bigcup_{i\in N}M_i^0$ for the set of chores that are zero-valued for \emph{some} agent, and $M^1:=M\setminus M^0$ for the chores valued $-1$ by \emph{all} agents.
We also need these sets restricted to a subset of \emph{active} agents $\overline{N}\subseteq N$:
\[
\begin{aligned}
M^1(\overline{N})&:=\{c\in M:\ v_i(c)=-1\ \forall i\in\overline{N}\},\\
M^0(\overline{N})&:=M\setminus M^1(\overline{N}).
\end{aligned}
\]

\paragraph{Algorithm description.}
The algorithm first partitions $M^1$, the chores disliked by everyone, as evenly as possible into $n$ bundles, each containing either $k=\lfloor|M^1|/n\rfloor$ or $k+1$ chores; we call bundles of size $k$ \emph{rich} and bundles of size $k+1$ \emph{poor}.
The remaining chores (in $M^0$) are then assigned while preserving the invariant that no bundle's value drops below $-(k+1)$ from any relevant agent's perspective.

To achieve this, the algorithm repeatedly runs the \textsc{Chores-Subroutine} (Algorithm~\ref{alg:chores-subroutine}): whenever an agent $i$ still has unallocated zero-valued chores and values some bundle $B_j$ at exactly $-k$, the subroutine adds \emph{all} currently unallocated chores of $M_i^0$ to $B_j$; agent $i$'s value for $B_j$ is unchanged (still $-k$), and a tracker $t_j$ records that agent $i$ was the last agent to alter $B_j$.
After the subroutine terminates, every remaining unallocated chore is disliked by all agents recorded in some tracker.
If enough chores remain, one chore is added to every rich bundle, which yields \EF{}+\EQ{} among the currently active agents and an \EFX{}+\EQX{} completion overall.
Otherwise, the poor bundles are finalized and the algorithm recurses on the remaining rich bundles and active agents, after returning to the unallocated pool all previously assigned chores that are not universally disliked within the reduced active set.

We first prove three observations about properties of the \textsc{Chores-Subroutine} and the partial allocation obtained from it.

\begin{observation}\label{lem:tracker-bijection}
Suppose that, after an execution of \textsc{Chores-Subroutine}, the unallocated set $U$ is nonempty. Then every rich bundle has a non-zero tracker. Moreover, distinct rich bundles have distinct trackers. Consequently, the trackers of the rich bundles and the rich bundles are in bijection, and the number of untracked active agents equals the number of poor bundles.
\end{observation}

\begin{algorithm}[htp]
\begin{algorithmic}[1]
\small
\REQUIRE Active agents $\overline{N}\subseteq N$, target values
$-k,-(k+1)$, bundles $\{B_j\}_{j\in \overline{N}}$, unallocated set $U$
\ENSURE Updated bundles $\{B_j\}_{j\in \overline{N}}$, trackers
$\{t_j\}_{j\in \overline{N}}$, and unallocated set $U$
\STATE Initialize trackers $t_j \gets 0$ for all $j\in\overline{N}$.
\WHILE{$\exists\, i,j\in \overline{N}$ s.t. $M_i^0\cap U\ne\emptyset$ and
$v_i(B_j)=-k$}
        \STATE $B_j \gets B_j \cup (M_i^0\cap U)$
        \STATE $U \gets U \setminus M_i^0$
        \STATE $t_j \gets i$
\ENDWHILE
\RETURN $\{B_j\}_{j\in \overline{N}}$, $\{t_j\}_{j\in \overline{N}}$, $U$
\end{algorithmic}
\caption{\textsc{Chores-Subroutine}}
\label{alg:chores-subroutine}
\end{algorithm}

\begin{algorithm}[htp]
\begin{algorithmic}[1]
\small
\REQUIRE A binary chores instance. 
\ENSURE An \EFX{}+\EQX{} allocation $A=(A_1,\ldots,A_n)$.
\STATE Let $k := \lfloor |M^1|/n \rfloor$.
\STATE Partition $M^1$ into bundles $B_1,\ldots,B_n$ with
$|B_j|\in\{k,k+1\}$ for all $j\in[n]$.
\STATE Unallocated set $U \gets M^0$;\quad active agents
$\overline{N}\gets N$.
\STATE Rich bundles $R \gets \{j\in \overline{N} : |B_j|=k\}$ and poor
bundles $P\gets \{j\in \overline{N} : |B_j|=k+1\}$.
\STATE $(\{B_j\},\{t_j\},U) \gets
\textsc{Chores-Subroutine}(\overline{N},k,\{B_j\},U)$
\WHILE{$U\ne\emptyset$}
    \STATE Consider untracked active agents
    $F:=\{i\in \overline{N}:\ \nexists j\ \text{ s.t. } t_j=i\}$; and arbitrarily assign them the poor bundles, i.e.,
    $\{A_i\}_{i\in F}\gets \{B_j\}_{j\in P}$ (by \Cref{lem:tracker-bijection}, $|F| = |P|$).
    \IF{\textbf{(Case 1)} $|U|\ge |R|$}
        \STATE Add one chore from $U$ to each rich bundle
        $\{B_j\}_{j\in R}$; remove these chores from $U$.
        \STATE Assign rich bundles as $A_i\gets B_j$ if $t_j=i$.
        \STATE Distribute each remaining $c\in U$ to any agent $i\in\overline{N}$ with $v_i(c)=0$, by adding $c$ to $A_i$; set $U\gets\emptyset$, and \textbf{return} $A=(A_1, \ldots, A_n)$.
    \ELSE
        \STATE \textbf{(Case 2)} \hfill\textcolor{gray}{// $0<|U|<|R|$}
        \STATE Finalize the agents in $F$ and redefine the active
        agents as $\overline{N}\gets \overline{N}\setminus F$.
        Also relabel each remaining rich bundle by its tracker: for each $j\in R$ with $t_j=i$, rename $B_j$ as $B_i$. 
        \item[] \hfill\textcolor{gray}{// Active agents and bundles now share the index set $\overline{N}$.}
        \STATE Pick any $|U|$ bundles in $\overline{N}$ and give each bundle
        one item from $U$; these are the new poor bundles $P$. Redefine
        $R\gets \overline{N}\setminus P$.
        \item[] \hfill \textcolor{gray}{// $U$ is empty at this stage, but
        the next for-loop returns some items to $U$.}
        \FORALL{$j\in \overline{N}$}
           \STATE Return to $U$ all items previously added to $B_j$ from
           $M^0(\overline{N})$, i.e.,
           $U \gets U \cup (B_j\setminus M^1(\overline{N}))$;\ \
           $B_j \gets B_j\cap M^1(\overline{N})$.
           \hfill \textcolor{gray}{// By Observation~\ref{sobs:chores-2},
           no item added in line 15 is returned to $U$.}
        \ENDFOR
        \STATE $(\{B_j\}_{j\in \overline{N}},\{t_j\}_{j\in \overline{N}},U)
              \gets \textsc{Chores-Subroutine}(\overline{N},k,\{B_j\},U)$
    \ENDIF
\ENDWHILE
\STATE \textbf{Final assignment:}
        For each $j$ with $t_j\ne 0$, set $A_{t_j}\gets B_j$.
        Assign the remaining bundles to the remaining agents via any
        bijection.
\RETURN $A=(A_1,\ldots,A_n)$
\end{algorithmic}
\caption{\EFX{}+\EQX{} for Binary Chores}
\label{alg:ex-post-chores}
\end{algorithm}

\begin{proof}
Distinct bundles have distinct non-zero trackers because an agent can alter at most one bundle in \textsc{Chores-Subroutine}: after agent $i$ acts once, all currently unallocated chores that she values at zero are removed from $U$.

It remains to show that every rich bundle is tracked. At the beginning of each call to \textsc{Chores-Subroutine}, every rich bundle consists of exactly $k$ chores that are disliked by every active agent.
Suppose some rich bundle $B_j$ remains untracked after the subroutine. Since $U\neq\emptyset$, choose $c\in U$. By construction $U\subseteq M^0(\overline{N})$, so some active agent $i$ has $v_i(c)=0$. Since $B_j$ is untracked, it was never modified and $v_i(B_j)=-k$. Thus the pair $i,j$ still satisfies the while-loop condition of \textsc{Chores-Subroutine}, contradicting termination.

Hence all rich bundles have distinct trackers. Since $|\overline{N}|=|R|+|P|$, exactly $|P|$ active agents are untracked.
\end{proof}

\begin{observation}\label{sobs:chores-1}
    Consider the bundles $\{B_j\}_{j\in \overline{N}}$ obtained after
    running the \textsc{Chores-Subroutine}.
    For every $i\in \overline{N}$ such that $M_i^0\cap U \ne \emptyset$, we
    have $v_i(B_j) \le -(k+1)$ for every bundle in
    $\{B_j\}_{j\in\overline{N}}$.
\end{observation}

\begin{proof}
    If there existed a bundle $B_j$ with $v_i(B_j) = -k$, then, by the condition of the while-loop in the subroutine, we would have $M_i^0\cap U = \emptyset$, a contradiction. (No bundle can have value above $-k$ for any agent, since the bundles start with at least $k$ universally disliked chores.)
\end{proof}

Next, notice that an agent $i$ alters at most one bundle in the \textsc{Chores-Subroutine}: after her first iteration, \emph{all} items that are zero-valued for $i$ (i.e., all of $M^0_i\cap U$) are assigned, and the while-loop condition is never again satisfied for $i$.
Thus, unless a tracker equals $0$, no two trackers are marked with the same agent.

\begin{observation}\label{sobs:chores-2}
   Consider the bundles $\{B_j\}_{j\in \overline{N}}$ obtained after running the \textsc{Chores-Subroutine}, and consider an agent $i\in \overline{N}$ such that $t_j=i$ for some $j$.
   Then $v_i(c)=-1$ for every $c\in U$.
\end{observation}

\begin{proof}
    Consider the iteration in which $t_j$ was set to $i$.
    In this iteration, all items in $U$ that are zero-valued for $i$ (i.e., $M_i^0\cap U$) are added to the bundle $B_j$.
    Thus every item remaining in $U$ afterwards has value $-1$ for $i$, and the subroutine never adds items back to $U$.
\end{proof}

\begin{proof}[Proof of \Cref{thm:binary_chores_efx_eqx}]
First, we show that each step in the algorithm is well-defined. Then, we argue termination and the two fairness guarantees in the output allocation.

\smallskip\noindent\textbf{Well-defined.}
Whenever $U\ne \emptyset$, \Cref{lem:tracker-bijection} shows that every rich bundle has a distinct tracker and that exactly $|P|$ active agents are untracked. Thus, the poor bundles can be assigned bijectively to the untracked agents. In Case~2, the remaining active agents are exactly the trackers of the rich bundles; thus, relabeling the rich bundles by the indices of new active agents is well-defined.

\smallskip\noindent\textbf{Termination.}
Each iteration of the while-loop in the \textsc{Chores-Subroutine} removes the nonempty set $M_i^0\cap U$ from $U$, strictly decreasing $|U|$; thus the subroutine terminates after at most $|M|$ iterations.
For the main algorithm, observe that whenever Case~1 (line~8) applies, $U$ is set to $\emptyset$ and the algorithm terminates.
It therefore suffices to bound the number of Case~2 executions.
In Case~2, the active agent set $\overline{N}$ is updated to $\overline{N}\setminus F$, and hence decreases by $|F| = |P|$. The only possible Case~2 execution with $P=\emptyset$ is the first one; since Case~2 has $|U|>0$, it creates a nonempty new set of poor bundles. Thus, after at most one Case~2 execution that does not decrease $|\overline{N}|$, every subsequent Case~2 execution strictly decreases it. Hence Case~2 is executed at most $n+1$ times. The algorithm therefore terminates after polynomially many iterations.

\smallskip\noindent\textbf{A structural observation.}
Every item allocated in line~15 belongs to $M^1(\overline{N})$.
Indeed, line~7 assigns the current poor bundles to agents $F=\{i\in\overline{N}:\nexists j \text{ s.t. } t_j=i\}$.
After the agents in $F$ are finalized and removed in Case~2, the remaining active agents are exactly those with $t_j=i$ for some $j$.
For each such agent $i$, Observation~\ref{sobs:chores-2} gives $v_i(c)=-1$ for all $c\in U$.
Every item of $U$ is therefore a chore of value $-1$ to every agent in the current $\overline{N}$, i.e., $U\subseteq M^1(\overline{N})$, so the items drawn from $U$ in line~15 lie in $M^1(\overline{N})$.

\smallskip\noindent\textbf{If the algorithm terminates in Case~1.}
Suppose the algorithm terminates through Case~1.
After one chore from $U$ is added to each rich bundle, every currently active agent has realized utility $-(k+1)$. Moreover, every active agent values every active bundle at most $-(k+1)$: for a tracked agent this follows from \Cref{sobs:chores-2}, while for an untracked agent who has a zero-valued chore in $U$ it follows from \Cref{sobs:chores-1}; if she has no such chore, she dislikes every chore added in Case~1.

Now consider an agent $i$ finalized in an earlier Case~2 execution.
Her bundle was a poor bundle and therefore consists of exactly $k+1$ chores, each valued $-1$ by $i$. Hence, for every $c\in A_i$, $ v_i(A_i\setminus\{c\})=-k$.
Every bundle contains at least the original $k$ chores from $M^1$, which are valued $-1$ by every agent, and therefore $ v_i(A_j)\le -k$ for every $i,j$.
Thus every previously finalized agent satisfies \EFX{} as well.
Moreover, every previously finalized poor bundle consists of $k+1$ chores that were disliked by every agent who remained active at the time it was finalized. Since the current active-agent set is a subset of that set, every currently active agent values every previously finalized bundle at $-(k+1)$.

Finally, the remaining chores distributed in line~11 are zero-valued to their recipients. Hence they do not change any recipient's realized utility and can only weakly decrease how other agents value that recipient's bundle. Consequently, all agents have realized utility $-(k+1)$, so \EQ{} (and hence \EQX{}) holds, and the \EFX{} guarantee established above is preserved. Therefore the final allocation is \EFX{}+\EQX{}.

For the rest of the proof, we assume that the algorithm terminates without executing Case~1.
\smallskip\noindent\textbf{\EQX{}.}
Suppose $v_i(A_i)=-k$ and $v_j(A_j)=-(k+1)$ for two agents $i,j$ (by construction all realized utilities lie in $\{-k,-(k+1)\}$).
We show that $A_j$ consists of exactly $k+1$ items, each of value $-1$ to $j$, which suffices to ensure $A$ satisfies \EQX{}.
Bundle $B_j$ became poor either in line~2 (initial partition of $M^1$) or by acquiring one item in line~15; by the structural observation, this item lies in $M^1(\overline{N})$ and hence has value $-1$ for $j$.
The for-loop in lines~17-19 removes every item not in $M^1(\overline{N})$---in particular every item that is zero-valued for its tracked recipient---from every active bundle before further items are added, and the items distributed in line~11 are always zero-valued for their recipients.
Hence $A_j$ consists of exactly $k+1$ items, each valued $-1$ by $j$, and for every $c\in A_j$, \[v_j(A_j\setminus\{c\})=-k=v_i(A_i),\]
establishing \EQX{}.

\smallskip\noindent\textbf{\EFX{}.}
Every final bundle contains at least $k$ chores from the original set $M^1$, so $ v_i(A_j)\le -k$ for all $i,j$.
Since every realized utility lies in $\{-k,-(k+1)\}$, if agent $i$ envies agent $j$, necessarily $v_i(A_i)=-(k+1)$ and $v_i(A_j)=-k$.
We show $|A_i|=k+1$ with every item of $A_i$ valued $-1$ by $i$, which suffices to ensure that $A$ satisfies \EFX{}.
The for-loop in lines~17-19 removes every zero-valued item from every active bundle, so $A_i$ consists either of $k+1$ items from line~2, or of $k$ items from line~2 plus one item from line~15.
The line-2 items lie in $M^1$ (value $-1$ to all agents), and by the structural observation the line-15 item, if present, lies in $M^1(\overline{N})$ (also value $-1$ to $i$).
Hence every item in $A_i$ has value $-1$ for $i$, and for every $c\in A_i$,
\[
v_i(A_i\setminus\{c\})=-k=v_i(A_j),
\]
establishing \EFX{}.
\end{proof}

\section{Ex-ante and Ex-post Guarantees}\label{sec:exante-expost}

In this section, we seek randomized allocations that are simultaneously \EF{}+\EQ{} ex-ante while retaining approximate fairness ex-post.

For two agents with normalized valuations, we obtain ex-post \EFone{} (\Cref{sec:prob_twoagents}).
These results establish a sharp contrast between goods and chores. For normalized binary goods, our general ex-post guarantee is \EQone{} (\Cref{sec:prob_goods}), and \Cref{cor:binary-113-lottery} rules out a universal strengthening to \EFone{}+\EQone{}. For chores, the ex-post guarantee can be strengthened to \EFone{}+\EQone{} without normalization, even for restricted-additive valuations (\Cref{sec:prob_chores}).

\subsection{Two-Agent Instances}
\label{sec:prob_twoagents}

The key is to ensure that the support of the randomized allocation consists only of \EFone{} allocations, regardless of which agent receives which bundle. We achieve this through a discrete version of exact division (consensus halving); for two agents, it always exists and can be computed in polynomial time \citep{kyropoulou2020almost}.

\begin{theorem}\label{prop:ex-ante-EQQ-ex-post-EF1-2-agents} 
    For two-agent normalized instances (goods or chores), an ex-ante \EF{}+\EQ{} and ex-post \EFone{} randomized allocation always exists.
\end{theorem}

\begin{proof}
We first consider goods.
By the Exact1 result of \citet{kyropoulou2020almost}, there exists a partition $\{X,Y\}$ of $M$ such that \emph{both} agents consider \emph{both} bundles to be \EFone{}; that is, for every agent $i\in\{1,2\}$ and every side $Z\in\{X,Y\}$, if agent $i$ receives $Z$ and $M\setminus Z\neq\emptyset$, then there is a good $g$ in the other side with $v_i(Z)\ge v_i((M\setminus Z)\setminus\{g\})$.
Consider the uniform distribution over the two allocations $(X,Y)$ and $(Y,X)$.
Every allocation in the support is \EFone{} by the choice of the partition, so the randomized allocation is ex-post \EFone{}.
For the ex-ante guarantees, each agent receives each side with probability $1/2$, so for both agents $i$ and both bundles $Z$,
\[
\mathbb{E}[v_i(\text{own bundle})]
=\frac{v_i(X)+v_i(Y)}{2}
=\frac{v_i(M)}{2},
\]
and the same computation applies to the other agent's bundle.
By normalization, $v_1(M)=v_2(M)$, so both agents obtain the same expected utility (ex-ante \EQ{}), and each agent's expected value for the other agent's bundle equals the expected value of her own (ex-ante \EF{}).

For chores, apply the goods argument to the nonnegative cost functions $d_i=-v_i$, which are additive and normalized.
The resulting partition $\{X,Y\}$ satisfies, for each agent $i$ and each side $Z$ with $M\setminus Z\neq\emptyset$,
\[
d_i(Z)\ \ge\ d_i\big((M\setminus Z)\setminus\{c\}\big) \quad\text{for some }c\in M\setminus Z.
\]
Written for the orientation in which agent $i$ receives the side $M\setminus Z$, this states precisely that there is a chore $c$ in her own bundle with $v_i\big((M\setminus Z)\setminus\{c\}\big)\ge v_i(Z)$, i.e., chores-\EFone{}.
As $Z$ ranges over both sides, both orientations are \EFone{} for both agents, and the uniform distribution over the two orientations is ex-post \EFone{} and, exactly as above, ex-ante \EF{}+\EQ{}.
\end{proof}
The ex-post guarantee in \Cref{prop:ex-ante-EQQ-ex-post-EF1-2-agents} cannot generally be strengthened to \EFX{}. See the counter-example in \Cref{sec:appendix-randomized}.

\subsection{Binary Goods}\label{sec:prob_goods}

For general additive valuations, agent-independent \EFone{} partitions do not extend beyond two agents, as such partitions may not exist for $n\geq3$.
We therefore turn to binary valuations. In this setting, an allocation satisfying ex-ante \EF{} and ex-post \EFone{} can be obtained by computing a fractional maximum Nash welfare (MNW) solution and implementing it as a randomized allocation supported on deterministic MNW allocations \citep{halpern2020fair}.

However, MNW solutions provide no equitability guarantees, either ex-ante or ex-post, since equitability may require sacrificing welfare. This is in contrast to MNW, which maximizes the geometric mean of agents' values, and therefore yields Pareto-optimal allocations.
For normalized binary goods, the uniform fractional allocation is simultaneously \EF{} and \EQ{} ex-ante. We show that it can be implemented as a distribution over deterministic allocations in which every agent receives either the floor or ceiling of her fractional approved utility; hence every realization is \EQone{}. The complete proof, based on total unimodularity, follows.

\begin{theorem}\label{thm:binary-lottery}
For normalized binary goods instances with an arbitrary number of agents, a randomized allocation that is ex-ante \EF{}+\EQ{} and ex-post \EQone{} always exists and can be computed in polynomial time.
\end{theorem}

\begin{proof}
Let $\Gamma_i=\{g\in M:v_i(g)=1\}$ be agent $i$'s approved goods.
By normalization, all approval sets have the same size; write $k:=|\Gamma_i|$ for this common value.
Set $\ell=\lfloor k/n\rfloor$ and $u=\lceil k/n\rceil$, and consider the polytope
\[
\mathcal{P}=\left\{x\ge0:\
\begin{aligned}
&\textstyle\sum_i x_{ig}=1 &&(g\in M),\\
&\textstyle\ell\le\sum_{g\in\Gamma_i}x_{ig}\le u &&(i\in N)
\end{aligned}
\right\}.
\]
The coefficient matrix of the displayed sums is the node-edge incidence matrix of a bipartite graph whose two sides are the goods and the agents, with additional unit columns: each variable $x_{ig}$ has a coefficient $1$ in the row of good $g$, and a coefficient $1$ in the row of agent $i$ if and only if $g\in\Gamma_i$ (columns of non-approved pairs have only their good-row entry).
Such matrices are totally unimodular, and all bounds are integral, so every vertex of $\mathcal{P}$ is integral.
An integral point of $\mathcal{P}$ assigns each good to exactly one agent and gives every agent either $\ell$ or $u$ approved goods; hence every vertex corresponds to a deterministic allocation whose realized utilities all lie in $\{\ell,u\}$.

The equal-share point $x^*$ with $x^*_{ig}=1/n$ for all $i,g$ lies in $\mathcal{P}$, because each approved degree equals $\sum_{g\in\Gamma_i}1/n=k/n\in[\ell,u]$.
A standard bipartite-flow decomposition expresses $x^*$, in polynomial time, as a convex combination of polynomially many integral vertices of $\mathcal{P}$.
Take this convex combination as the randomized allocation.

Because the decomposition preserves $x^*$, every agent receives every good with probability exactly $1/n$.
Thus, from any agent $i$'s perspective, every random bundle (her own and every other agent's) has expected value $k/n$; hence the randomized allocation is ex-ante \EF{}, and the expected realized utilities are all equal to $k/n$, so it is also ex-ante \EQ{}.

Finally, consider any allocation in the support and two agents $i,j$.
Realized utilities lie in $\{\ell,u\}$ with $u-\ell\le1$.
If $v_i(A_i)\ge v_j(A_j)$, the \EQone{} condition for the pair $(i,j)$ holds after removing any good (or vacuously if $A_j=\emptyset$).
Otherwise $v_i(A_i)=\ell$ and $v_j(A_j)=u=\ell+1$; then $A_j$ contains a good $g$ approved by $j$, and $v_j(A_j\setminus\{g\})=u-1=\ell\le v_i(A_i)$.
Hence every support allocation is \EQone{}.
\end{proof}

As a corollary of \Cref{thm:binary-113-counterexample}, it is easy to see that the ex-post \EQone{} guarantee of \Cref{thm:binary-lottery} cannot be universally strengthened to \EFone{}+\EQone{} for normalized binary goods.

\begin{corollary}
\label{cor:binary-113-lottery}
There is a normalized binary goods instance with 113 agents for which no randomized allocation is ex-post \EFone{}+\EQone{}, even without any ex-ante requirements.
\end{corollary}
\begin{proof}
Every allocation in the support would have to satisfy \EFone{}+\EQone{}, but \Cref{thm:binary-113-counterexample} admits no such allocation.
\end{proof}

This impossibility does not affect the laminar case that we will study next.

\paragraph{Laminar binary goods.}
Under the laminar structure of \Cref{thm:laminar_binary_goods}, however, the ex-post guarantee \emph{can} be strengthened to \EFone{}+\EQone{}.
As in the construction of \Cref{thm:binary-lottery}, ex-ante equitability again forces us to hand goods to agents who do not value them.
Under laminarity, however, any two approval sets are identical or disjoint, so we can cut every approval set into the same number of balanced blocks, let each group of identical agents rotate cyclically through the blocks of its own set, and give the surplus blocks to fixed outside agents, who value them at zero; see \Cref{sec:appendix-randomized-laminar-goods} for the complete proof.

\begin{theorem}\label{thm:laminar-binary-lottery}
For normalized binary goods instances with laminar approval sets, an ex-ante \EF{}+\EQ{} and ex-post \EFone{}+\EQone{} randomized allocation always exists and can be computed in polynomial time, even when the instance has an arbitrary number of agents.
\end{theorem}

In \Cref{sec:expost-limits} we show that the ex-post guarantee cannot be strengthened to \EFX{}, even for laminar approval sets and even when we do not impose any ex-post equitability requirement or ex-ante envy requirement.

\subsection{Chores: Escaping Normalization for Binary Valuations and Beyond}
\label{sec:prob_chores}

For chores, stronger guarantees are available ex-post, and they require neither normalization nor any bound on the number of agents.
This was established by \citet{sun2025randomized} through their group-strategyproof (in-expectation) mechanism \textsc{RandChore}, which for $1$-restricted additive costs---precisely the restricted-additive chores of our model---returns a randomized allocation that is ex-ante \EF{}, \EQ{} and \textsc{PROP} and ex-post \EFone{}, \EQone{} and \textsc{PROP1}, while additionally being ex-ante and ex-post \PO{}.
Restricting their guarantee to the two notions we study yields the following statement.

\begin{theorem}[\citealp{sun2025randomized}, Theorem~4.4]\label{thm:chores-binary-ex-ante-ex-post}
    For restricted-additive chores instances, not necessarily normalized, an ex-ante \EF{}+\EQ{} and ex-post \EFone{}+\EQone{} randomized allocation can be computed in polynomial time for an arbitrary number of agents.
\end{theorem}

For completeness, we record a short self-contained proof of the statement in this form.
The construction below can be seen as a partial derandomization of \textsc{RandChore}: it fixes, rather than randomizes, the assignment of the chores that some agent values at zero, and it rotates a single round-robin partition instead of drawing a fresh permutation.

\begin{proof}
    Let $M^0$ contain the chores that some agent values at zero, and  
    assign each $c\in M^0$ permanently to such an agent; let $Z_1,\ldots,Z_n$ be the resulting allocation.  
    Every remaining chore is in $M^1:=M\setminus M^0$ and has the same value $v(c)<0$ for every agent.
    We now construct a common \EFone{}+\EQone{} partition of $M^1$, so that a uniform random assignment of this partition among the agents leads to an ex-ante \EF{}+\EQ{} randomized allocation with the desired ex-post guarantees.

    Since agents are identical on $M^1$, every \EFone{} allocation is also \EQone{}. 
    Moreover, any \EFone{} allocation $(B_1, \ldots, B_n)$ on $M^1$ is such that all bundles $B_p$ are \EFone{} for all agents since the valuations are identical; one such allocation can be computed via round robin.

Now we construct a randomized allocation supported on $n$ \EFone{}+\EQone{} allocations as follows: For $r=0,\ldots,n-1$, define $\sigma_r(i)=1+((i-1+r)\bmod n)$ and
\[
 A_i^{(r)}=B_{\sigma_r(i)}\cup Z_i.
\]
The uniform distribution over these $n$ allocations gives each agent every $B_p$ exactly once. 
Therefore $\mathbb E[v_i(A_i)]=v(M^1)/n$ for every $i$, proving ex-ante \EQ{}.  For
$i\ne j$,
\[
\mathbb E[v_i(A_j)]
=\frac{v(M^1)}n+v_i(Z_j)
\le\frac{v(M^1)}n
=\mathbb E[v_i(A_i)],
\]
which proves ex-ante \EF{}.
\end{proof}

In \Cref{sec:expost-limits} we show that the ex-post guarantee is tight, i.e. we present a binary chores instance with two agents and two chores in which no ex-ante \EF{}+\EQ{} randomized allocation is ex-post \EFX{} or ex-post \EQX{}. 
The counter-example, however, is not normalized, and an interesting open problem is whether an ex-ante \EF{}+\EQ{} and ex-post \EFX{}+\EQX{} allocation always exists for binary normalized chores. 

Since every allocation in the support of the randomized allocation constructed in \Cref{thm:chores-binary-ex-ante-ex-post} is simultaneously \EFone{} and \EQone{}, we obtain the following deterministic guarantee as an immediate consequence. 

\begin{corollary}\label{cor:rest-add-chores-ef1-eq1}
    For unnormalized, restricted-additive chores instances, an \EFone{}+\EQone{} allocation always exists, and such allocations can be computed in polynomial time.
\end{corollary}

\section{Concluding Remarks}
\label{sec:conclusion}

Our work studies the simultaneous satisfaction of fundamentally different fairness notions, within and across the deterministic and randomized settings, revealing a sharp distinction between goods and chores. Normalized binary goods admit \EFone{}+\EQone{} allocations for at most seven agents, but our 113-agent counterexample shows that simultaneous existence fails in general. The central remaining open question is to narrow the gap: find polynomial-time solutions for larger $n$ or find counter-examples for smaller $n$.
Conceptually, our work also initiates the study of cross-notion best-of-both-worlds fairness beyond envy and equitability, where ex-ante guarantees for one fairness notion are combined with ex-post guarantees for another.

\section*{AI Usage Disclosure}
The proof for \Cref{thm:binary-7} was obtained with the help of GPT-5.6 and Claude-Fable 5. 
The authors first established the result for $n=3$ agents themselves, and subsequently extended it to $n=5$ and then $n=7$ agents with AI assistance.
The 113-agent counterexample in \Cref{thm:binary-113-counterexample} was found by GPT-6 Astra.
Additionally, the authors used GPT-5.6 for editorial suggestions and writing the draft. 
All mathematical arguments and the final text were reviewed and verified by the authors.

\section*{Acknowledgments}
The authors acknowledge the National Science Foundation (NSF) for financial support. HH and SP were supported through CAREER Awards IIS-2144413 and IIS-2107173, and LX and CZ were supported by Awards 2450124, 2517733, and 2518373.
% The authors also acknowledge GPT-6 Astra for finding the normalized binary counterexample in \Cref{thm:binary-113-counterexample}.

\bibliographystyle{plainnat}
\bibliography{refs}

\appendix

\section{Omitted Material from Section 3}

\subsection{Fundamental Barriers}\label{sec:appendix_challenges}
\begin{proof}[Complete proof of \Cref{prop:non-exist-3-and-more}]
Consider again the instance depicted in the proof sketch above: there are $n\ge 3$ agents and $m=2n$ goods; for every $i\in[n-1]$, $v_i(g_1)=1$ and $v_i(g_j)=\varepsilon$ for all $j\ge 2$, where $\varepsilon>0$ is sufficiently small so that $(2n-1)\varepsilon<\delta$ with
\[\delta:=\frac{1+(m-1)\varepsilon}{m}=\frac{1+(2n-1)\varepsilon}{2n},\]
and agent $n$ values every good at $\delta$.
Every agent values the grand bundle at $1+(2n-1)\varepsilon$, so the instance is normalized.

Consider any allocation $A$.
Since $n\ge 3$, there exists an agent $i\in[n-1]$ who does not receive $g_1$; hence $v_i(A_i)\le(2n-1)\varepsilon<\delta$.
We claim that agent $n$ can receive at most one good.
Otherwise, after removing any good from $A_n$, agent $n$'s remaining utility is at least $\delta$, contradicting \EQone{} with respect to $i$.

Thus, at least $2n-1$ goods are allocated among the first $n-1$ agents.
By the pigeonhole principle, one of these agents, say $j$, receives at least three goods.
Agent $n$ values her own bundle at most $\delta$, whereas, for every $g\in A_j$, we have $v_n(A_j\setminus\{g\})\ge 2\delta>\delta\ge v_n(A_n)$.
Hence agent $n$ \EFone{}-envies agent $j$, establishing that no allocation is simultaneously \EFone{}+\EQone{}.
\end{proof}

An analogous construction rules out the chores setting.

The following proposition is the chores analogue of \Cref{prop:non-exist-3-and-more}.

\begin{proposition}\label{sprop:non-exist-chores}
For $n\ge 3$ agents with normalized bivalued valuations, an allocation of chores that is both \EFone{} and \EQone{} may fail to exist.
\end{proposition}

\begin{proof}
We first provide an instance for every $n\ge 4$ and then a separate one for $n=3$.

For $n\ge 4$, let $m=2n$.
The first $n-2$ agents (the \emph{flat} agents) have identical disutilities
\[v_i(c)=-\delta:=-\frac{1}{2n}\]
for every chore $c$.
The last two agents, $n-1$ and $n$, assign disutility $-(1-(2n-1)\varepsilon)$ to chore $c_1$ (the \emph{heavy} chore) and disutility $-\varepsilon$ to every other chore (the \emph{light} chores), where $\varepsilon>0$ is chosen so that $(2n-1)\varepsilon<\delta$.
Every agent has total disutility $-1$, so the instance is normalized, and every valuation is bivalued.

Suppose some flat agent receives at least two chores.
After removing any one chore from her bundle, its remaining disutility is at most $-\delta$.
Since the heavy chore $c_1$ can be assigned to at most one of agents $n-1$ and $n$, the other of these two agents receives only light chores and thus incurs cost at most $(2n-1)\varepsilon<\delta$.
Therefore \EQone{} is violated between the flat agent and this agent.
Hence every one of the first $n-2$ agents receives at most one chore.

Consequently, the flat agents together receive at most $n-2$ chores, leaving at least $2n-(n-2)=n+2$ chores for agents $n-1$ and $n$.
By the pigeonhole principle, one of these two agents receives at least $\lceil (n+2)/2\rceil\ge 3$ chores.
Even after removing any single chore, her bundle still contains at least two chores, each costing her at least $\varepsilon$, so its remaining cost is at least $2\varepsilon$.

Since at most one flat agent can receive the heavy chore $c_1$ and there are $n-2\ge 2$ flat agents, some flat agent holds no heavy chore and at most one light chore, so agents $n-1$ and $n$ value her bundle at cost at most $\varepsilon$ (possibly $0$ if the bundle is empty).
Thus the overloaded agent has cost at least $2\varepsilon$ after removing any one chore, while she values this flat agent's bundle at cost at most $\varepsilon$, contradicting \EFone{}.
Hence no allocation simultaneously satisfies \EFone{} and \EQone{}.

Now consider $n=3$.
Fix $0<\varepsilon<1/12$ and take eight chores with the following disutilities:
\[
\begin{array}{c|cc}
 & c_1 & c_j\ (j=2,\ldots,8)\\
\hline
\text{Agent } 1 & -12.5 & -12.5\\
\text{Agent } 2 & -(86-7\varepsilon) & -(2+\varepsilon)\\
\text{Agent } 3 & -93 & -1
\end{array}
\]
Every agent's total disutility is $-100$, so the instance is normalized and bivalued.

Agent~1 cannot receive two or more chores: after removing any one of them, her remaining disutility is at most $-12.5$, whereas one of agents~2 and~3 does not receive the heavy chore $c_1$; since agent~1 holds at least two of the eight chores, that agent holds at most six light chores and incurs cost at most $6(2+\varepsilon)<12.5$ (as $\varepsilon<1/12$), violating \EQone{}.

Moreover, agent~1 cannot receive zero chores.
Otherwise, all eight chores are distributed between agents~2 and~3, so one of them receives at least four chores.
After removing her most costly chore, that agent still incurs strictly positive cost from at least three chores, while agent~1's bundle is empty and has cost $0$, an \EFone{} violation.
Therefore agent~1 receives exactly one chore.

If agent~1 receives a light chore instead of $c_1$, the remaining seven chores are divided between agents~2 and~3, and one of them receives at least four chores.
After removing her most costly chore, the remaining at least three light chores still have strictly greater cost (to her) than agent~1's single light chore, violating \EFone{}.
Hence necessarily $A_1=\{c_1\}$.

The remaining seven light chores must then be split between agents~2 and~3.
If some agent receives at least five of them, she \EFone{}-envies the other even after removing one chore, so one agent receives three chores and the other four.
In either split, agent~3's cost is at most $4$, whereas agent~2's cost after removing one chore is at least $2(2+\varepsilon)=4+2\varepsilon>4$, contradicting \EQone{}.
Therefore, no allocation simultaneously satisfies \EFone{} and \EQone{} for $n=3$.
\end{proof}

\subsection{Computational Barriers}\label{sec:appendix_hardness}

We prove a more general statement, for which we need the ``up to $k$ goods'' versions of our fairness notions.

\begin{definition}\label{sdef:efk-eqk}
For an integer $k\ge 1$, an allocation $A$ is \emph{envy-free up to $k$ goods} (\EF{}-$k$) if, for every pair of agents $i,j\in N$, there exists a subset $S\subseteq A_j$ with $|S|\le k$ such that $v_i(A_i)\ge v_i(A_j\setminus S)$.
Similarly, $A$ is \emph{equitable up to $k$ goods} (\EQ{}-$k$) if, for every pair $i,j\in N$, there exists a subset $S\subseteq A_j$ with $|S|\le k$ such that $v_i(A_i)\ge v_j(A_j\setminus S)$.
\end{definition}

For $k=1$ these notions coincide with \EFone{} and \EQone{}, so the following theorem contains \Cref{thm:hardness} as the special case $k=1$.

\begin{theorem}\label{sthm:hardness-general}
For two agents with additive (possibly unnormalized) valuations and any fixed integer $k\ge 1$, deciding whether an allocation that is both \EF{}-$k$ and \EQ{}-$k$ exists is \NPH{}.
\end{theorem}

\begin{proof}
We reduce from the classical \textsc{Partition} problem: given positive integers $x_1,\ldots,x_n$, decide whether they can be partitioned into two subsets with equal sum.
Let $T=\sum_{i=1}^n x_i$; we may assume $T$ is even, as otherwise the instance is trivially a NO-instance.

\smallskip\noindent\textbf{Construction.}
Create a fair-division instance with two agents $A$ and $B$ and the following goods:
\begin{itemize}
  \item $n$ \emph{ordinary} goods $o_1,\dots,o_n$;
  \item $3k$ \emph{dummy} goods $d_1,\dots,d_{3k}$.
\end{itemize}
The additive valuations are
\[
\begin{array}{c|cc}
 & \text{ordinary } o_i & \text{dummy } d_t \\\hline
A & x_i & D\\
B & x_i & 0
\end{array}
\]
where $D$ is any integer with $D>T$ (e.g., $D=T+1$, which has polynomial size).
We claim that the instance admits a complete \EF{}-$k$+\EQ{}-$k$ allocation if and only if the \textsc{Partition} instance is solvable.

\smallskip\noindent\textbf{($\Rightarrow$) From a fair allocation to a partition.}
Let $t_A$ and $t_B$ denote the numbers of dummies allocated to $A$ and $B$, so $t_A+t_B=3k$.

\smallskip
\emph{Claim 1: $t_A\le k$.}
Suppose $t_A>k$ and consider \EQ{}-$k$ for the ordered pair $(B,A)$.
After removing any set of at most $k$ goods from $A$'s bundle, at least one dummy remains, so the remaining value (to $A$) is at least $D>T$.
But $B$'s realized utility is at most $T$.
Hence $v_B(B)\ge v_A(A\setminus S)$ fails for every admissible $S$, contradicting \EQ{}-$k$.

\smallskip
\emph{Claim 2: $t_B\le t_A+k$.}
Suppose $t_B>t_A+k$ and consider \EF{}-$k$ for the ordered pair $(A,B)$.
After removing any set of at most $k$ goods from $B$'s bundle, $B$ still retains at least $t_B-k>t_A$ dummies, so $A$ values the remaining bundle at least $(t_B-k)D\ge(t_A+1)D=t_AD+D>t_AD+T$.
On the other hand, $A$'s own total value is at most $t_AD+T$.
Hence $A$ still envies $B$ after every admissible removal, contradicting
\EF{}-$k$.

\smallskip
Combining Claims 1 and 2 with $t_A+t_B=3k$ yields $t_A=k$ and $t_B=2k$.
Now let $S_A$ and $S_B$ be the sets of ordinary goods assigned to $A$ and $B$, and write $s_A=v_A(S_A)$ and $s_B=v_B(S_B)$, so $s_A+s_B=T$.

\smallskip
\emph{Claim 3: $s_A\le s_B$.}
Suppose $s_A>s_B$ and consider \EQ{}-$k$ for the pair $(B,A)$.
The most value-reducing removal of at most $k$ goods from $A$'s bundle removes her $k$ dummies, leaving value $s_A>s_B=v_B(B)$; any other removal leaves at least one dummy and value at least $D>T\ge s_B$.
Either way \EQ{}-$k$ fails.

\smallskip
\emph{Claim 4: $s_A\ge s_B$.}
Suppose $s_A<s_B$ and consider \EF{}-$k$ for the pair $(A,B)$.
The most value-reducing removal (from $A$'s perspective) of at most $k$ goods from $B$'s bundle removes $k$ dummies, leaving value $kD+s_B>kD+s_A=v_A(A)$; any other removal leaves at least $k+1$ dummies and value at least $(k+1)D>kD+T\geq kD+s_A$.
Either way $A$ still envies $B$, contradicting \EF{}-$k$.

\smallskip
Claims 3 and 4 yield $s_A=s_B=T/2$, so the ordinary goods split into two sets of equal sum, i.e., the \textsc{Partition} instance is a YES-instance.

\smallskip\noindent\textbf{($\Leftarrow$) From a partition to a fair allocation.}
Suppose $S_A,S_B$ partition the ordinary goods with equal sums $T/2$.
Allocate to $A$ the goods $S_A$ plus any $k$ dummies, and to $B$ the goods $S_B$ plus the remaining $2k$ dummies.
Then $v_A(A)=T/2+kD$ and $v_B(B)=T/2$.
For \EF{}-$k$: agent $B$ does not envy $A$ (she values $A$'s bundle at $T/2=v_B(B)$), and removing $k$ dummies from $B$'s bundle leaves value $T/2+kD$ to agent $A$, eliminating her envy.
For \EQ{}-$k$: removing the $k$ dummies from $A$'s bundle leaves both realized utilities equal to $T/2$; all other ordered comparisons are immediate.
Thus the allocation is \EF{}-$k$+\EQ{}-$k$.
The construction is computable in polynomial time, which completes the reduction.
\end{proof}

\subsection{Limits of Strengthening \EFone{}+\EQone{}}

The main text notes that, even under normalization and binary valuations, neither of the two guarantees in \EFone{}+\EQone{} can generally be strengthened.
The following construction makes this precise.

\begin{proposition}\label{sprop:six-agent-both}
For six agents with normalized binary valuations, an allocation of goods that is \EFone{}+\EQX{}, and likewise one that is \EFX{}+\EQone{}, may fail to exist.
\end{proposition}

Both parts use the same instance; the two claims are established as Propositions~\ref{sprop:six-agent-ef1-eqx}~and~\ref{sprop:six-agent-efx-eq1} below.
Let $n=6$ and $m=18$, and let $L$ and $R$ be disjoint sets of nine goods.
Agent~1 approves exactly the goods in $L$:
\[
v_1=(\underbrace{1,\ldots,1}_{L},\underbrace{0,\ldots,0}_{R}),
\]
while every agent $p\in\{2,\ldots,6\}$ approves exactly the goods in $R$:
\[
v_p=(\underbrace{0,\ldots,0}_{L},\underbrace{1,\ldots,1}_{R}).
\]
All valuations are binary and normalized.
For an allocation $A$, write
\[
\begin{aligned}
x&=|A_1\cap L|,& y&=|A_1\cap R|,\\
a_p&=|A_p\cap L|,& b_p&=|A_p\cap R|
\end{aligned}
\]
for $p\in\{2,\ldots,6\}$.
The realized utilities are $u_1=x$ and $u_p=b_p$, and
\begin{equation}
x+\sum_{p=2}^6 a_p=9,
\qquad
y+\sum_{p=2}^6 b_p=9.
\label{eq:six-agent-counts}
\end{equation}

\begin{proposition}
\label{sprop:six-agent-ef1-eqx}
The above instance admits no \EFone{}+\EQX{} allocation.
\end{proposition}

\begin{proof}
Suppose, for contradiction, that an allocation $A$ satisfies both \EFone{} and \EQX{}.
Let $\mu=\min\{x,b_2,\ldots,b_6\}$ denote the minimum realized utility.
Since utilities are integral, \EQX{} implies that every agent has utility either $\mu$ or $\mu+1$.
Moreover, if an agent holds a good that she values at $0$, removing that good does not change the value of her bundle; hence such an agent must have utility exactly $\mu$.
Consequently, for every $p\in\{2,\ldots,6\}$,
\begin{equation}
y>0 \Longrightarrow x=\mu,
\qquad
a_p>0 \Longrightarrow b_p=\mu.
\label{eq:six-agent-zero-rule}
\end{equation}
The \EFone{} comparisons of agent~1 against agent $p$, and of agent $p$
against agent~1, give
\begin{equation}
a_p\le x+1,
\qquad
y\le b_p+1
\qquad (p\in\{2,\ldots,6\}).
\label{eq:six-agent-ef1}
\end{equation}

\smallskip\noindent\textbf{Case 1: $y=0$.}
Then $\sum_{p=2}^6 b_p=9$ by \eqref{eq:six-agent-counts}.
If $\mu=0$, then every utility is at most $1$, so $\sum_p b_p\le 5<9$, which is impossible; hence $\mu\ge 1$. Since every $b_p\in\{\mu,\mu+1\}$ and $\sum_p b_p=9$, we get $\mu=1$ and $(b_2,\ldots,b_6)$ consists of four $2$'s and one $1$, and $x\in\{1,2\}$.
By \eqref{eq:six-agent-zero-rule}, the four agents with utility $2$ cannot hold any goods from $L$.
Therefore all $9-x\ge 7$ goods from $L$ not assigned to agent~1 belong to the unique agent with utility $1$, i.e., $a_p\ge 7>x+1$ for that agent, contradicting \eqref{eq:six-agent-ef1}.

\smallskip\noindent\textbf{Case 2: $y>0$.}
Then $x=\mu$ by \eqref{eq:six-agent-zero-rule}.
If $x=0$, then every $b_p\in\{0,1\}$, so $9-y=\sum_p b_p\le 5$ and hence $y\ge 4$, contradicting $y\le b_p+1\le 2$ from \eqref{eq:six-agent-ef1}.
Moreover, $5\mu\le\sum_p b_p=9-y\le 8$ gives $\mu\le 1$; therefore $x=\mu=1$, and every $b_p\in\{1,2\}$.
Let $h$ be the number of agents with $b_p=2$.
Then $5+h=\sum_p b_p=9-y$, so $h=4-y$, and $5-h=1+y$ agents have utility $1$.
By \eqref{eq:six-agent-zero-rule}, only these $1+y$ agents may receive goods from $L$, and by \eqref{eq:six-agent-ef1} each receives at most $x+1=2$ of them.
Since they must receive all $9-x=8$ goods of $L$ not held by agent~1, we get $2(1+y)\ge 8$, i.e., $y\ge 3$.
But a utility-$1$ agent exists, and \eqref{eq:six-agent-ef1} requires $y\le b_p+1=2$ for her, a contradiction.
\end{proof}

\begin{proposition}\label{sprop:six-agent-efx-eq1}
The same instance admits no \EFX{}+\EQone{} allocation.
\end{proposition}

\begin{proof}
Suppose an allocation $A$ satisfies both \EFX{} and \EQone{}, and let $\mu$ be the minimum realized utility.
Under \EQone{}, all utilities lie in $\{\mu,\mu+1\}$.
Since $\sum_{p}b_p\le 9$ over five agents, $\mu\le 1$.

Suppose $\mu=0$.
Then $x,b_p\le1$ for all $p$, so $9-y=\sum_p b_p\le 5$ and $y\ge4$.
If $x>0$, then $A_1$ contains an $L$-good, which every agent $p$ values at zero; \EFX{} from a utility-$0$ agent $p$ toward agent~1 would require $b_p\ge v_p(A_1\setminus\{\text{$L$-good}\})=y\ge 4$, which is impossible.
Thus $x=0$, and removing any $R$-good from $A_1$ leaves value $y-1\ge3$ to agent $p$, still above $b_p\le1$, again contradicting \EFX{}.
Hence $\mu=1$, so $x,b_p\in\{1,2\}$.

Every bundle $A_p$ contains an $R$-good, which agent~1 values at zero.
Consequently, \EFX{} from agent~1 toward $p$ requires $a_p=v_1(A_p\setminus\{\text{$R$-good}\})\le x$ for every $p$.
Since $\sum_pa_p=9-x$, the case $x=1$ would give $9-x=8>5\ge\sum_p a_p$, a contradiction; therefore $x=2$.

Since $x=2$, the bundle $A_1$ contains an $L$-good, which every agent $p$ values at zero, so \EFX{} from $p$ toward agent~1 requires $b_p\ge y$ for every $p$.
Summing gives $9-y\ge5y$, hence $y\le1$.
If $y=0$, the values $(b_2,\ldots,b_6)$ are four $2$'s and one $1$; if $y=1$, they are three $2$'s and two $1$'s.
A utility-$2$ agent cannot receive any $L$-good: otherwise a utility-$1$ agent $q$ would still see both of that agent's $R$-goods after removing the zero-valued $L$-good, i.e., $v_q(A_p\setminus\{\text{$L$-good}\})\ge 2>1=b_q$, violating \EFX{}.
Thus all $9-x=7$ goods of $L$ not held by agent~1 must go to the one or two utility-$1$ agents, whose combined capacity is at most $2x=4<7$ by the bound $a_p\le x$.
This final contradiction completes the proof.
\end{proof}

\subsection{Two-Agent Instances}\label{sec:appendix_two_agent_deterministic}

\paragraph{Existence for goods.}
Assume without loss of generality that the common grand-bundle value is positive and scale it to one; if $v_1(M)=v_2(M)=0$, then all goods are zero-valued for both agents (values are nonnegative) and any allocation is trivially \EFX{}+\EQX{}.
Following \citet{plaut2020almost}, a \emph{leximin++ allocation} is one that (i) maximizes the minimum realized utility, (ii) subject to that, maximizes the number of items held by an agent with the minimum utility, and (iii) subject to that, maximizes the larger utility.

\begin{proposition}
\label{sprop:leximin-two-agents}
For two agents with normalized valuations, every leximin++ allocation of goods is both \EFX{} and \EQX{}.
\end{proposition}

\begin{proof}
Let $A=(A_1,A_2)$ be a leximin++ allocation and relabel the agents so that $v_1(A_1)\le v_2(A_2)$.

\smallskip\noindent\emph{\EQX{}.}
Suppose \EQX{} fails.
Then there exists $g\in A_2$ with $v_1(A_1)<v_2(A_2\setminus\{g\})$.
Transfer $g$ to agent~1, i.e., consider $A'=(A_1\cup\{g\},A_2\setminus\{g\})$. Agent~2's new utility remains strictly above the old minimum $v_1(A_1)$.
If $v_1(g)>0$, agent~1's utility strictly increases, so the minimum utility strictly increases.
If $v_1(g)=0$, the minimum utility is unchanged while the minimum-utility agent holds strictly more items.
Both cases contradict leximin++ optimality, so $A$ is \EQX{}.

\smallskip\noindent\emph{\EFX{}.}
First, the two agents cannot envy each other simultaneously: swapping their bundles would strictly increase both realized utilities, contradicting leximin++ optimality.
If agent~2 envied agent~1 while agent~1 did not envy agent~2, normalization would give $v_2(A_2)<1/2$ and $v_1(A_1)\ge1/2$, contradicting $v_1(A_1)\le v_2(A_2)$. Hence any envy is from agent~1 toward agent~2, in which case $v_1(A_1)<1/2$.

Suppose \EFX{} fails, i.e., there is $g\in A_2$ with $v_1(A_1)<v_1(A_2\setminus\{g\})$.
Consider the two parts $X=A_1\cup\{g\}$ and $Y=A_2\setminus\{g\}$, and let agent~2 choose her preferred part, with agent~1 receiving the other.
Agent~2's new utility is at least $\max\{v_2(X),v_2(Y)\}\ge v_2(M)/2=1/2$.
Agent~1 receives either $Y$, which she values strictly above $A_1$, or $X$, which weakly improves her utility and, in case of a tie ($v_1(g)=0$), strictly increases her number of items.
In every case the leximin++ objective strictly improves, a contradiction.
Hence $A$ is \EFX{}.
\end{proof}

\begin{remark}
Although computing a leximin++ allocation is computationally hard in general, it can be computed in polynomial time for binary valuations.
\end{remark}

\paragraph{Non-existence for chores.}
We now show that the guarantee of Proposition~\ref{sprop:leximin-two-agents} does not carry over to chores: under the all-item convention for \EFX{}/\EQX{}, even two agents with normalized valuations may admit no \EFX{}+\EQX{} allocation.

\begin{proposition}
\label{sprop:chores-efx-eqx-nonexist}
For two agents with normalized valuations, an allocation of chores that is both \EFX{} and \EQX{} may fail to exist, already with three chores.
\end{proposition}

\begin{proof}
Consider three chores $c_1,c_2,c_3$ and the valuations
\begin{center}
\begin{tabular}{c|ccc}
 & $c_1$ & $c_2$ & $c_3$\\
\hline
Agent $1$ & $0$ & $-1$ & $-2$\\
Agent $2$ & $-2$ & $-1$ & $0$
\end{tabular}
\end{center}
Both agents value the grand bundle at $-3$, so the instance is normalized.

The instance is symmetric under simultaneously exchanging the two agents and the chores $c_1\leftrightarrow c_3$ (keeping $c_2$ fixed): this relabeling maps each valuation function onto the other, and it maps an allocation $(A_1,A_2)$ to an allocation whose bundles have the same values to their owners and to the other agent.
Hence an allocation is \EFX{} (resp.\ \EQX{}) if and only if its image is, and it suffices to rule out the allocations with $|A_1|\le 1$; the cases $|A_1|\ge2$ are their images under the symmetry.

\smallskip\noindent\emph{Case $A_1=\emptyset$.}
Then $A_2=\{c_1,c_2,c_3\}$ with $v_2(A_2)=-3$, while $v_1(A_1)=0$.
Removing the zero-valued chore $c_3$ from $A_2$ leaves $v_2(A_2\setminus\{c_3\})=-3<0=v_1(A_1)$, so \EQX{} fails.

\smallskip\noindent\emph{Case $A_1=\{c_1\}$.}
Then $v_1(A_1)=0$ and $A_2=\{c_2,c_3\}$ with $v_2(A_2)=-1$.
Removing the zero-valued chore $c_3$ from $A_2$ leaves $v_2(A_2\setminus\{c_3\})=-1<0=v_1(A_1)$, so \EQX{} fails.

\smallskip\noindent\emph{Case $A_1=\{c_2\}$.}
Then $v_1(A_1)=-1$ and $A_2=\{c_1,c_3\}$ with $v_2(A_2)=-2$.
Removing $c_3$ from $A_2$ leaves $v_2(A_2\setminus\{c_3\})=-2<-1=v_1(A_1)$, so \EQX{} fails.

\smallskip\noindent\emph{Case $A_1=\{c_3\}$.}
Then $A_2=\{c_1,c_2\}$, and agent~2 values agent~1's bundle at $v_2(A_1)=v_2(\{c_3\})=0$, while $v_2(A_2)=-3$, so agent~2 envies agent~1.
Removing $c_2$ (agent~2's least costly chore in her bundle) leaves $v_2(A_2\setminus\{c_2\})=-2<0=v_2(A_1)$, so the envy persists after the removal of some chore, and \EFX{} fails.

\smallskip
In every case, \EFX{} or \EQX{} is violated; by the symmetry noted above, the same holds for all remaining allocations.
Hence no allocation of this instance is simultaneously \EFX{} and \EQX{}.
\end{proof}

\begin{remark}
The violations above hinge on zero-valued chores and the all-item convention adopted in this paper.
Under the weaker convention in which only negatively valued chores may be removed, the instance of
Proposition~\ref{sprop:chores-efx-eqx-nonexist} does admit an \EFX{}+\EQX{} allocation, e.g., $A_1=\{c_1\}$ and $A_2=\{c_2,c_3\}$.
\end{remark}

The algorithm for chores mirrors \textsc{Greedy-Balance} for goods (\Cref{alg:two_agent_goods}), with the roles of the two agents adjusted for disutilities: the next chore is assigned to the agent with the currently \emph{higher} utility (i.e., smaller burden), and each agent takes a remaining chore with the largest value difference in her favor.

\begin{algorithm}[htp]
\caption{\textsc{Greedy-Balance} for two-agent chores}
\label{alg:utility-balance-chores}
\textbf{Input}: A normalized two-agent chores instance\\
\textbf{Output}: An allocation $(A_1,A_2)$
\begin{algorithmic}[1]
\STATE Initialize $A_1,A_2\leftarrow\emptyset$, $U_1,U_2\leftarrow 0$, and
$R\leftarrow M$.
\WHILE{$R\neq\emptyset$}
    \IF{$U_1\ge U_2$}
        \STATE Choose $c^\star\in\arg\max_{c\in R}\{v_1(c)-v_2(c)\}$.
        \STATE $A_1\leftarrow A_1\cup\{c^\star\}$ and
        $U_1\leftarrow U_1+v_1(c^\star)$.
    \ELSE
        \STATE Choose $c^\star\in\arg\max_{c\in R}\{v_2(c)-v_1(c)\}$.
        \STATE $A_2\leftarrow A_2\cup\{c^\star\}$ and
        $U_2\leftarrow U_2+v_2(c^\star)$.
    \ENDIF
    \STATE $R\leftarrow R\setminus\{c^\star\}$.
\ENDWHILE
\RETURN $(A_1,A_2)$.
\end{algorithmic}
\end{algorithm}

The cross-bundle dominance invariant of \Cref{lem:cross_bundle_dominance} holds verbatim; its proof is independent of the signs of the values.

\begin{lemma}
\label{slem:aw-chores}
At every stage of Algorithm~\ref{alg:utility-balance-chores},
\[
v_1(A_1)\ge v_2(A_1) \qquad\text{and}\qquad v_2(A_2)\ge v_1(A_2).
\]
\end{lemma}

\begin{proof}
We prove the first inequality; the second is symmetric.
Define $\Delta(c)=v_1(c)-v_2(c)$; normalization gives $\Delta(M)=0$.
Suppose that at some stage $\Delta(A_1)=v_1(A_1)-v_2(A_1)<0$.
Since $\Delta(M)=0$, there exist $c\in A_1$ with $\Delta(c)<0$ and $d\notin A_1$ with $\Delta(d)>0$.
Consider the iteration in which $c$ was assigned to agent~1.
At that moment $d$ must already have been allocated: otherwise the algorithm, which assigns to agent~1 a remaining chore maximizing $\Delta(\cdot)$, would have chosen $d$ instead of $c$.
Thus $d$ was assigned earlier to agent~2.
At that earlier iteration $c$ was still available with
$\Delta(c)<0<\Delta(d)$; since agent~2 receives a remaining chore \emph{minimizing} $\Delta(\cdot)$, the algorithm would have chosen $c$ rather than $d$, a contradiction.
\end{proof}

\begin{proof}[Proof of \Cref{cor:two_agents_chores}]
Let $A=(A_1,A_2)$ be the returned allocation.

\smallskip\noindent\emph{\EQone{}.}
Without loss of generality, assume $v_1(A_1)\le v_2(A_2)$.
The ordered comparison from agent~2 is immediate: since values are nonpositive, removing any chore from $A_2$ only increases its value, so $v_2(A_2\setminus\{c\})\ge v_2(A_2)\ge v_1(A_1)$ for every $c\in A_2$ (and the condition is vacuous if $A_2=\emptyset$). For the comparison from agent~1, if $A_1=\emptyset$ then $0=v_1(A_1)\le v_2(A_2)\le 0$ forces $v_2(A_2)=0$ and the condition holds trivially.
Otherwise, let $c_\ell$ be the \emph{last} chore assigned to agent~1, at some iteration $t_\ell$, and let $A_1(t_\ell),A_2(t_\ell)$ denote the bundles held immediately before that assignment.
Since the algorithm assigned $c_\ell$ to agent~1, we have $v_1(A_1(t_\ell))=U_1\ge U_2=v_2(A_2(t_\ell))$ at that moment.
Moreover $A_1\setminus\{c_\ell\}=A_1(t_\ell)$, and every chore allocated after iteration $t_\ell$ went to agent~2, which can only decrease her utility, so $v_2(A_2)\le v_2(A_2(t_\ell))$.
Combining,
\[
v_1(A_1\setminus\{c_\ell\}) = v_1(A_1(t_\ell)) \ge v_2(A_2(t_\ell)) \ge v_2(A_2),
\]
which is the \EQone{} condition for agent~1.

\smallskip\noindent\emph{\EFone{}.}
Suppose agent~1 envies agent~2, i.e., $v_1(A_1)<v_1(A_2)$.
By Lemma~\ref{slem:aw-chores}, $v_2(A_2)\ge v_1(A_2)$, hence $v_1(A_1)<v_2(A_2)$, and in particular $A_1\neq\emptyset$.
The \EQone{} argument above then yields a chore $c_\ell\in A_1$ with $v_1(A_1\setminus\{c_\ell\})\ge v_2(A_2)\ge v_1(A_2)$, which is exactly the \EFone{} condition.
The case in which agent~2 envies agent~1 is symmetric.
The algorithm performs one greedy step per chore and thus runs in polynomial time.
\end{proof}

\subsection{Why the Barrier at Seven Agents?} \label{sec:appendix_7_agent_barriers}

\begin{remark}
\label{srem:beyond-seven}
The analysis of \Cref{alg:level-scatter} is \emph{local}: when a residual good with valuer set $T$ ($|T|=q$) cannot be scattered, every candidate bundle already holds $t+1$ goods valued by some member of $T$, so by pigeonhole one valuer sees at least $(t+1)\lceil(n-q)/q\rceil$ goods in other bundles---contradicting the mass bound $k\le r(t+1)-1$ precisely when $\lceil(n-q)/q\rceil\ge r-1$ (\Cref{lem:small-population-seven}).  The number of blocking configurations available to an adversary grows with $s\cdot r\le\lfloor n^{2}/4\rfloor$, while the saturation mass that the normalized budgets can afford grows only linearly, as $2n-2$; since $\lfloor n^{2}/4\rfloor\le 2n-2$ holds exactly for $n\le7$---with equality at $n=7$---the counting argument closes at seven with zero slack and first fails at $n=8$ (already for $q=2$).

\begin{figure*}[t]
\centering
\scriptsize
\setlength{\tabcolsep}{3.4pt}
\renewcommand{\arraystretch}{1.15}
\begin{tabular}{r|cc|cc|cc|cc|cc|cc|cc|c||cc}
 & \multicolumn{2}{c|}{$P_u$} & \multicolumn{2}{c|}{$P_v$}
 & \multicolumn{2}{c|}{$P_{j_1}$} & \multicolumn{2}{c|}{$P_{j_2}$}
 & \multicolumn{2}{c|}{$P_{j_3}$} & \multicolumn{2}{c|}{$P_{j_4}$}
 & \multicolumn{2}{c|}{$P_{j_5}$} & \multicolumn{1}{c||}{$P_{j^{\dagger}}$}
 & \multicolumn{2}{c}{$R$}\\
 & $z^{u}_{1}$ & $z^{u}_{2}$ & $z^{v}_{1}$ & $z^{v}_{2}$
 & $o_{1}$ & $o_{2}$ & $o_{3}$ & $o_{4}$ & $o_{5}$ & $o_{6}$
 & $o_{7}$ & $o_{8}$ & $o_{9}$ & $o_{10}$ & $o_{11}$
 & $g_{1}$ & $g_{2}$\\
\hline
$u$ & $\mathbf{1}$ & $\mathbf{1}$ & 0 & 0
 & \cellcolor{red!18}1 & \cellcolor{red!18}1
 & \cellcolor{red!18}1 & \cellcolor{red!18}1
 & \cellcolor{red!18}1 & \cellcolor{red!18}1
 & 0 & 0 & 0 & 0
 & \cellcolor{yellow!35}1
 & 1 & 1\\
$v$ & 0 & 0 & $\mathbf{1}$ & $\mathbf{1}$
 & \cellcolor{blue!15}1 & \cellcolor{blue!15}1
 & 0 & 0 & 0 & 0
 & \cellcolor{blue!15}1 & \cellcolor{blue!15}1
 & \cellcolor{blue!15}1 & \cellcolor{blue!15}1
 & \cellcolor{yellow!35}1
 & 1 & 1\\
\hline
$j_1$ & 0 & 0 & 0 & 0 & $\mathbf{1}$ & $\mathbf{1}$ & 1 & 1 & 1 & 1 & 1 & 1 & 1 & 1 & 1 & 0 & 0\\
$j_2$ & 0 & 0 & 0 & 0 & 1 & 1 & $\mathbf{1}$ & $\mathbf{1}$ & 1 & 1 & 1 & 1 & 1 & 1 & 1 & 0 & 0\\
$j_3$ & 0 & 0 & 0 & 0 & 1 & 1 & 1 & 1 & $\mathbf{1}$ & $\mathbf{1}$ & 1 & 1 & 1 & 1 & 1 & 0 & 0\\
$j_4$ & 0 & 0 & 0 & 0 & 1 & 1 & 1 & 1 & 1 & 1 & $\mathbf{1}$ & $\mathbf{1}$ & 1 & 1 & 1 & 0 & 0\\
$j_5$ & 0 & 0 & 0 & 0 & 1 & 1 & 1 & 1 & 1 & 1 & 1 & 1 & $\mathbf{1}$ & $\mathbf{1}$ & 1 & 0 & 0\\
$j^{\dagger}$ & 0 & 0 & 0 & 0 & 1 & 1 & 1 & 1 & 1 & 1 & 1 & 1 & 1 & 1 & $\mathbf{1}$ & 0 & 0\\
\hline\hline
$c_{u,\cdot}$ & \multicolumn{2}{c|}{$2$} & \multicolumn{2}{c|}{$0$}
 & \multicolumn{2}{c|}{\cellcolor{red!18}$\;2=t{+}1\;$}
 & \multicolumn{2}{c|}{\cellcolor{red!18}$\;2=t{+}1\;$}
 & \multicolumn{2}{c|}{\cellcolor{red!18}$\;2=t{+}1\;$}
 & \multicolumn{2}{c|}{$0$} & \multicolumn{2}{c|}{$0$}
 & \cellcolor{yellow!35}$1$
 & \multicolumn{2}{c}{}\\
$c_{v,\cdot}$ & \multicolumn{2}{c|}{$0$} & \multicolumn{2}{c|}{$2$}
 & \multicolumn{2}{c|}{\cellcolor{blue!15}$\;2=t{+}1\;$}
 & \multicolumn{2}{c|}{$0$} & \multicolumn{2}{c|}{$0$}
 & \multicolumn{2}{c|}{\cellcolor{blue!15}$\;2=t{+}1\;$}
 & \multicolumn{2}{c|}{\cellcolor{blue!15}$\;2=t{+}1\;$}
 & \cellcolor{yellow!35}$1$
 & \multicolumn{2}{c}{}\\
\end{tabular}
\caption{An eight-agent instance ($m=17$, $k=11$) on which the scattering phase stalls from the displayed maximum partial allocation.  Columns are grouped by bundle; bold entries mark each good's owner.  The six outside agents $j_1,\dots,j_5,j^{\dagger}$ collectively value only $o_1,\dots,o_{11}$, so the egalitarian level is $t=1$, every bundle holds $t+1=2$ goods except $P_{j^{\dagger}}$, and the goods $g_1,g_2$ (type $\{u,v\}$) are residual.  The two bottom rows give the counts $c_{u,P_j}:=v_u(P_j)$ and $c_{v,P_j}:=v_v(P_j)$: every bundle except $P_{j^{\dagger}}$ is saturated for $u$ (red) or for $v$ (blue), in an anti-correlated pattern---$u$'s slack sits at $P_{j_4},P_{j_5}$, $v$'s at $P_{j_2},P_{j_3}$---while a $\{u,v\}$-good needs room in \emph{both} rows at the same bundle.  Only $P_{j^{\dagger}}$ (yellow) offers such a slot, and only once: after $g_1$ enters it, $g_2$ cannot be scattered, although every per-agent count is far below its budget.  An \EFone{}+\EQone{} allocation exists nonetheless (assign $g_1,g_2$ to $u$, releasing the freely scatterable singletons $z^{u}_{1},z^{u}_{2}$). Notice that the barrier stems from the \emph{choice} of partial allocation, not the instance.}
\label{fig:barrier-instance}
\end{figure*}
 
Importantly, beyond seven agents, scattering no longer succeeds from \emph{every} maximum partial allocation.  Consider $n=8$ (\Cref{fig:barrier-instance}): agents $u,v$ and $j_1,\dots,j_5,j^{\dagger}$, and $m=17$ goods with $k=11$: goods $z^{u}_{1},z^{u}_{2}$ valued by $\{u\}$; $z^{v}_{1},z^{v}_{2}$ by $\{v\}$; $g_{1},g_{2}$ by $\{u,v\}$; and $o_{1},\dots,o_{11}$, each valued by all six outside agents, where $o_{1},o_{2},o_{11}$ are additionally valued by both $u$ and $v$, $o_{3},\dots,o_{6}$ by $u$, and $o_{7},\dots,o_{10}$ by $v$.  The outside agents collectively value only eleven goods, so $t=1$.  In the maximum partial allocation that gives the $z$-goods to their valuers, $o_{1,2}$ to $j_{1}$, $o_{3,4}$ to $j_{2}$, $o_{5,6}$ to $j_{3}$, $o_{7,8}$ to $j_{4}$, $o_{9,10}$ to $j_{5}$, and $o_{11}$ to $j^{\dagger}$, the residual goods are $g_{1},g_{2}$, and every bundle except $P_{j^{\dagger}}$ is saturated for $u$ or for $v$ in an \emph{anti-correlated} pattern: $u$'s slack sits at $j_{4},j_{5}$, $v$'s at $j_{2},j_{3}$, yet a $\{u,v\}$-good needs room in both rows \emph{at the same bundle}, which only $j^{\dagger}$ offers---once.  After $g_{1}$ enters $j^{\dagger}$, the good $g_{2}$ cannot be placed and the scattering phase stalls, even though every per-agent count is good; here the swap phase rescues the run (the residual type equals $S=\{u,v\}$, so $g_{1},g_{2}$ can be swapped into $u$'s bundle, releasing freely scatterable singletons).  

The example shows that beyond seven agents, the scattering phase need not succeed from every maximum partial allocation, even when an \EFone{}+\EQone{} allocation exists. It is not a nonexistence example. Extending the approach to more agents may require stronger structural transformations before scattering or a more careful choice of the initial maximum partial allocation. In view of the 113-agent nonexistence result in \Cref{thm:binary-113-counterexample}, no such rule can work on all normalized binary instances. The relevant challenge is to improve the population bound within $8,\ldots,112$, or to handle broader structural classes.
\end{remark}

\subsection{Generalizing the Binary Goods Counterexample} \label{sec:counter-example-generalization}

\begin{lemma}[A construction principle]
\label{lem:binary-cover-counterexample}
Suppose $s$ binary valuation types on $m$ goods each approve exactly $k$ goods. Let $B\ge1$ be an integer, and suppose every set of $B+2$ distinct goods is jointly approved by at least one type. Then there is a normalized instance with no complete \EFone{}+\EQone{} allocation and
\[
    n=s+\left\lfloor\frac{k}{B}\right\rfloor
    \qquad\text{whenever}\qquad
    m>(B+1)\left(s+\left\lfloor\frac{k}{B}\right\rfloor\right).
\]
\end{lemma}

\begin{proof}
Use one agent of each type and add $\lfloor k/B\rfloor$ copies of one chosen type. The $r=\lfloor k/B\rfloor+1$ agents of that type share an approval set of size $k<rB$, so at least one has utility at most $B-1$. In an \EQone{} allocation every agent then has utility at most $B$. The covering assumption and \EFone{} prohibit any bundle containing $B+2$ goods, since a common approving agent would value it above $B+1$. Completeness would therefore imply $m\le(B+1)n$, contrary to the assumed inequality.
\end{proof}

\Cref{thm:binary-113-counterexample} uses $s=31$, $B=2$, $k=165$, and $m=341$, so $n=31+\lfloor165/2\rfloor=113$. The ancillary file \texttt{ef1\_eq1\_113\_certificate.py} provides a standalone structural check: it generates the instance, verifies all approval counts and type multiplicities, and checks all $\binom{34}{4}=46{,}376$ four-label multisets, including repeated labels, for a common approving type. It verifies the capacity shortfall $341-339=2$; it does not enumerate allocations or establish minimality. The proof above is self-contained and does not depend on this check.

\section{A Simpler Algorithm for Five Agents}
\label{sec:five-supp}

For five agents, the swap phase and the use of non-valuers inside $S$ are not needed: the following simpler version of \textsc{Level-and-Scatter} assigns every residual good to an agent outside $S$.
We include it because both the algorithm and its case analysis are simpler than those required for \Cref{thm:binary-7}.

\begin{algorithm}[htp]
\caption{\textsc{Simple-Level-and-Scatter} for five agents}
\label{alg:level-scatter-five}
\begin{algorithmic}[1]
\small
\REQUIRE A normalized binary goods instance with five agents
\ENSURE An allocation $(A_1,\ldots,A_5)$
\STATE Remove all junk goods from $M$ and store them in $J$.
\STATE Compute the maximum feasible egalitarian level $t$.
\STATE Compute a maximum partial allocation $P=(P_1,\ldots,P_5)$ such that
every good in $P_i$ is valued by $i$ and $t\le |P_i|\le t+1$ for every $i$.
\STATE Let $R\gets M\setminus\bigcup_{i\in N}P_i$.
\STATE Construct the residual alternating graph and let $S$ be the set of
agents reachable from $R$.
\STATE Initialize $A_i\gets P_i$ for every $i\in N$.
\FOR{each $g\in R$}
    \STATE Choose $j\in N\setminus S$ such that
    $v_i(A_j\cup\{g\})\le t+1$ for every $i\in S$.
    \STATE $A_j\gets A_j\cup\{g\}$.
\ENDFOR
\STATE Distribute the goods in $J$ arbitrarily.
\RETURN $(A_1,\ldots,A_5)$.
\end{algorithmic}
\end{algorithm}

\begin{theorem}
\label{sthm:five-agents}
For five agents with normalized binary valuations, Algorithm~\ref{alg:level-scatter-five} computes, in polynomial time, an allocation of goods satisfying both \EFone{} and \EQone{}.
\end{theorem}

\begin{proof}
Let $k$ be the common number of goods valued by each agent.
If $R$ is empty, every utility is $t$ or $t+1$, and each partial bundle $P_j$ has cardinality at most $t+1$.
Adding junk goods does not affect any value.
The allocation is therefore \EQone{}, and $v_i(A_j)\le t+1\le v_i(A_i)+1$ for every $i,j$, which implies \EFone{} under binary valuations.

Suppose that $R\neq\emptyset$.
Lemma~\ref{lem:outsideS} gives $|N\setminus S|\ge2$, so $|S|\in\{1,2,3\}$.
We show that the scattering step is feasible in all three cases.

\smallskip\noindent
\emph{Case 1: $|S|=3$.}
Let $N\setminus S=\{p,q\}$.
Lemma~\ref{lem:closure}(3), together with the fact that at least one outside agent receives only $t$ goods, gives
\[
    k\le |P_p|+|P_q|\le2t+1.
\]
Each $i\in S$ therefore values at most $k-(t+1)\le t$ goods outside $P_i$, including all residual goods.
Assigning any residual good to either $p$ or $q$ consequently keeps both outside bundles worth at most $t+1$ to every agent in $S$.

\smallskip\noindent
\emph{Case 2: $|S|=1$.}
Write $S=\{i\}$.
Every residual good is valued only by $i$, and Lemma~\ref{lem:closure}(3) gives
\[
    k\le\sum_{j\notin S}|P_j|\le4t+3.
\]
Thus, outside $P_i$, there are at most $k-(t+1)\le3t+2$ goods valued by $i$.
The four outside bundles have combined capacity $4(t+1)$ under the cap $t+1$.
If no outside bundle could receive the next $i$-valued residual good, all four bundles would already have value $t+1$ for $i$, requiring $4(t+1)>3t+2$ such goods outside $P_i$, a contradiction.

\smallskip\noindent
\emph{Case 3: $|S|=2$.}
Write $S=\{1,2\}$.
The three outside bundles contain at most $3t+2$ goods in total, so for each
$i\in S$,
\[
    k-(t+1)\le2t+1.
\]
If a residual good is valued by only one agent $i\in S$, failure would mean that all three outside bundles already have value $t+1$ for $i$, requiring $3(t+1)>2t+1$ goods valued by $i$ outside $P_i$, a contradiction.
If the good is valued by both agents and no outside bundle can receive it, each of the three outside bundles is saturated for agent~1 or agent~2.
By the pigeonhole principle, one of these agents saturates at least two bundles and therefore values at least $2(t+1)=2t+2$ goods outside her own bundle, contradicting the bound $2t+1$.

\smallskip
Thus every residual good can be assigned. By Lemma~\ref{lem:closure}(2), agents outside $S$ value every residual good at zero, so no realized utility changes and all utilities remain in $\{t,t+1\}$; hence \EQone{} holds.
The construction also keeps every bundle worth at most $t+1$ to every agent, so $v_i(A_j)\le v_i(A_i)+1$ for all $i,j$.
As in the proof of \Cref{thm:binary-7}, binary valuations make this inequality equivalent to \EFone{}.
The flow computation, graph search, and scattering steps are all polynomial-time operations.
\end{proof}

\section{An Alternative Argument for Three Agents}
\label{sec:three-supp}

For three agents, we record a self-contained argument based on Hall's theorem; it is subsumed by \Cref{thm:binary-7} but uses considerably lighter machinery.

\begin{theorem}
\label{sthm:three-agents}
For three agents with normalized binary valuations, an allocation of goods satisfying both \EFone{} and \EQone{} always exists.
\end{theorem}

\begin{proof}
Let $N=\{1,2,3\}$ be the agents and $M$ the goods.
Since the valuations are binary and normalized, there exists an integer $t$ such that every agent values exactly $t$ goods at $1$.
We may assume that every good is valued by at least one agent: goods valued at $0$ by all agents can be set aside and distributed arbitrarily at the end, since they affect neither utilities nor \EFone{}/\EQone{}.

We first construct a maximal balanced partial allocation.
Let $u$ be the largest integer such that there exists a partial allocation in which every agent receives exactly $u$ goods that she values at $1$; equivalently, a maximum balanced $b$-matching in the bipartite graph between agents and goods.

We claim $u\ge\lfloor t/3\rfloor$.
Suppose no balanced partial allocation of size $k$ exists.
By Hall's condition for $b$-matchings, there exists a nonempty subset of agents $X$ such that $|N(X)|<k|X|$, where $N(X)$ is the set of goods valued by agents in $X$.
Since every agent values exactly $t$ goods, $|N(X)|\ge t$, and since $|X|\le3$, it follows that $t\le|N(X)|<3k$, i.e., $k>t/3$.
Hence a balanced partial allocation of size $\lfloor t/3\rfloor$ always exists.

Let $(A_1,A_2,A_3)$ be a maximal balanced partial allocation; each bundle contains exactly $u$ goods valued by its owner, so $v_i(A_j)\le|A_j|=u$ for all $i,j\in N$, and the partial allocation is both \EF{} and \EQ{}.
Let $S$ be the set of unallocated goods, and for each agent $i$ let $d_i=|\{g\in S:v_i(g)=1\}|$.
Without loss of generality assume $d_1\le d_2\le d_3$.
We complete the allocation by case analysis; throughout, ``capacity'' arguments refer to keeping $v_i(A_j)\le u+1$ for all $i,j$, which, together with all realized utilities lying in $\{u,u+1\}$, implies \EQone{} and (under binary valuations) \EFone{} exactly as in the proof of \Cref{thm:binary-7}.

\smallskip\noindent
\textbf{Case 1: $d_1\ge 3$.}
Every agent has at least three valued goods in $S$, so by Hall's theorem there is a system of distinct representatives assigning one additional valued good to each agent.
This yields a balanced partial allocation of size $u+1$, contradicting the maximality of $u$.

\smallskip\noindent
\textbf{Case 2: $d_1=2$.}
If $d_3\ge 3$, Hall's theorem again yields a balanced partial allocation of size $u+1$, a contradiction; hence $d_1=d_2=d_3=2$.
The unique barrier to a system of distinct representatives occurs when $|S|=2$ and all three agents value both remaining goods.
In this case, allocate one remaining good to agent~1 and the other to agent~2.
The realized utilities become $(u+1,u+1,u)$, so \EQone{} holds; moreover every agent values every bundle at most $u+1$ while realizing at least $u$, so \EFone{} holds as well.

\smallskip\noindent
\textbf{Case 3: $d_1=1$.}
If $d_2\ge 3$, or $d_2=2$ and $d_3\ge3$, Hall's theorem yields a contradiction with maximality.
Two subcases remain.

\smallskip
\emph{Subcase 3.1: $d_2=d_3=2$.}
The only obstruction occurs when $|S|=2$, agent~1 values exactly one remaining good, and agents~2 and~3 value both.
Allocate the good valued by agent~1 to agent~1 and the other good to agent~2.
The utilities are $(u+1,u+1,u)$ and the same reasoning as in Case~2 applies.

\smallskip
\emph{Subcase 3.2: $d_2=1$.}
First allocate one valued good from $S$ to agent~1; if agent~2 still has a valued unallocated good, allocate one such good to agent~2.
All remaining goods are valued only by agent~3.
We distribute them between $A_1$ and $A_2$ so that $v_3(A_1)\le u+1$ and $v_3(A_2)\le u+1$.
Such a distribution exists: agent~3 values exactly $t-u$ goods outside $A_3$, including those already allocated to $A_1$ and $A_2$, and currently $v_3(A_1),v_3(A_2)\le u+1$. Since $u\ge\lfloor t/3\rfloor$ implies $t-u\le2(u+1)$, the combined remaining capacity under these caps is sufficient for all remaining goods, so a greedy distribution respects both caps.
Every agent's realized utility lies in $\{u,u+1\}$, which gives \EQone{}, and every bundle is worth at most $u+1$ to every agent, which gives \EFone{}.

\smallskip\noindent
\textbf{Case 4: $d_1=0$.}

\smallskip
\emph{Subcase 4.1: $d_2\ge 2$ and $d_3\ge 3$.}
We derive a contradiction with the maximality of $u$.
If agent~2's bundle contains a good valued by agent~1, reassign that good to agent~1, then assign two valued goods from $S$ to agent~2 and one valued good from $S$ to agent~3.
Otherwise agent~3's bundle must contain a good valued by agent~1 (agent~1 values $t\ge u+1$ goods, hence some valued good lies outside $A_1$; since $d_1=0$, it lies in $A_2\cup A_3$): reassign it to agent~1, then assign one valued good from $S$ to agent~2 and two valued goods from $S$ to agent~3.
In either case every agent's utility increases to $u+1$, and by Hall's theorem the required distinct goods in $S$ exist because $d_2\ge2$ and $d_3\ge3$.
This contradicts maximality.

\smallskip
\emph{Subcase 4.2: $d_2=d_3=2$.}
If $|S|=2$, allocate one remaining good to agent~2 and the other to agent~3; the utilities become $(u,u+1,u+1)$, which satisfies \EQone{} and \EFone{} as before.
If $|S|\ge3$, arguments analogous to Subcase~4.1 again produce a balanced partial allocation of size $u+1$, contradicting maximality.

\smallskip
\emph{Subcase 4.3: $d_2=1$.}
Allocate one valued good to agent~2.
The remaining goods are valued only by agent~3; distribute them between $A_1$ and $A_2$ subject to $v_3(A_1),v_3(A_2)\le u+1$, which is possible by the same capacity argument as in Subcase~3.2.
All utilities lie in $\{u,u+1\}$ and all bundles are worth at most $u+1$ to every agent, so \EQone{} and \EFone{} hold.

\smallskip
\emph{Subcase 4.4: $d_2=0$.}
All remaining goods are valued only by agent~3; distribute them between $A_1$ and $A_2$ as in Subcase~4.3.
The same argument applies.

\smallskip
In every case we obtain a complete allocation satisfying \EFone{} and \EQone{}.
\end{proof}

\section[Omitted Material from Ex-ante and Ex-post Guarantees]{Omitted Material from \Cref{sec:exante-expost}} \label{sec:appendix-randomized}

\begin{proposition}\label{prop:no-ex-ante-EQ-ex-post-EFX}
    For two agents with normalized valuations, a randomized allocation that is ex-ante \EQ{} and ex-post \EFX{} may fail to exist, even in goods instances.
\end{proposition}

\begin{proof}
Take three goods with
\[
\begin{aligned}
(v_1(g_1),v_1(g_2),v_1(g_3))&=(0.1,\,0.5,\,0.4),\\
(v_2(g_1),v_2(g_2),v_2(g_3))&=(1,\,0,\,0).
\end{aligned}
\]
Both grand-bundle values equal one, so the instance is normalized.

We first show that in any \EFX{} allocation, $g_1$ must go to agent~2.
Suppose instead that $g_1\in A_1$.
If $A_1$ also contained $g_2$ or $g_3$, then $A_1$ would contain a good that agent~2 values at zero; removing it would leave agent~2's value for $A_1$ at least $1>v_2(A_2)=0$, so agent~2's envy would persist, violating \EFX{}.
If $A_1=\{g_1\}$, then agent~1's utility is $0.1$, while $v_1(A_2\setminus\{g_2\})=0.4$ and $v_1(A_2\setminus\{g_3\})=0.5$, so agent~1's envy survives every single-good removal, again violating \EFX{}.

Next, $g_2$ must go to agent~1: if $A_2\supseteq\{g_1,g_2\}$, then agent~1's own value is at most $v_1(g_3)=0.4$, while removing $g_1$ from $A_2$ leaves value at least $v_1(g_2)=0.5$, so \EFX{} fails.

Consequently, the only \EFX{} allocations are
\[
A=(\{g_2\},\{g_1,g_3\})
\qquad\text{and}\qquad
A'=(\{g_2,g_3\},\{g_1\}),
\]
and both are indeed \EFX{}.
Their realized utility pairs are $(0.5,1)$ and $(0.9,1)$, respectively.
Every randomized allocation supported on these allocations gives agent~1 expected realized utility at most $0.9$, strictly below agent~2's expected realized utility of $1$.
Hence no randomized allocation supported on \EFX{} allocations is ex-ante \EQ{}.
\end{proof}

\subsection{Randomized Guarantees for Laminar Binary Goods Instances}\label{sec:appendix-randomized-laminar-goods}

\paragraph{Setup.}
Since the valuations are binary, $v_i(M)=|\Gamma_i|$ for every agent $i$, so normalization means that all approval sets have the same cardinality; write $k:=|\Gamma_i|$, a single constant shared by all agents.
Binary valuations are restricted additive, so \Cref{sobs:laminar-types} applies: any two agents have either identical or disjoint approval sets.
Let $\Gamma^{(1)},\ldots,\Gamma^{(T)}$ be the distinct approval sets---so they are pairwise disjoint and $|\Gamma^{(\ell)}|=k$ for every $\ell$---and let $N_\ell:=\{i\in N: \Gamma_i=\Gamma^{(\ell)}\}$ be the corresponding group of identical agents, $n_\ell:=|N_\ell|$, so that $\sum_\ell n_\ell=n$.
Let
\[
    J:=M\setminus\bigcup\nolimits_{\ell\in[T]}\Gamma^{(\ell)}
\]
denote the goods that no agent values; thus $M$ is the disjoint union of $J$ and $\Gamma^{(1)},\ldots,\Gamma^{(T)}$, and $v_i(g)=0$ for every $g\in J$ and every $i$.
Finally, set
\begin{gather*}
    r:=\max_{\ell\in[T]} n_\ell,
    \qquad
    u:=\left\lfloor \frac{k}{r}\right\rfloor,\\
    s:=k-ru\in\{0,1,\ldots,r-1\},
\end{gather*}
and note that $r\le\sum_\ell n_\ell=n$.
All part indices below live in $\mathbb{Z}_r$, i.e., they are taken modulo $r$.

\paragraph{The construction.}
\begin{description}
\item[Step 1 (parts).] For every $\ell\in[T]$, fix an ordered partition of $\Gamma^{(\ell)}$ into $r$ pairwise disjoint, \emph{possibly empty} blocks $B^\ell_0,\ldots,B^\ell_{r-1}$ of which exactly $s$ have size $u+1$ and the remaining $r-s$ have size $u$.
This is feasible because $s(u+1)+(r-s)u=ru+s=k$.
(When $k<r$ we have $u=0$ and $r-k$ of the blocks are empty.)
Note that \emph{every} approval set is cut into the same number $r$ of blocks, including those of groups with $n_\ell<r$.

\item[Step 2 (labels).] For every $\ell$, fix an injective map
$\sigma_\ell:N_\ell\to\mathbb{Z}_r$ and put $S_\ell:=\sigma_\ell(N_\ell)$, so $|S_\ell|=n_\ell\le r$.

\item[Step 3 (fillers).] For every $\ell$, fix an injective map
\[
    \varphi_\ell:\ \mathbb{Z}_r\setminus S_\ell\ \longrightarrow\ N\setminus N_\ell .
\]
Such a map exists: $|\mathbb{Z}_r\setminus S_\ell|=r-n_\ell$, $|N\setminus N_\ell|=n-n_\ell$, and $r\le n$ gives $r-n_\ell\le n-n_\ell$. (If $T=1$, then $r=n_1=n$ and $S_1=\mathbb{Z}_r$, so $\varphi_1$ is the empty map; and whenever $r-n_\ell\ge1$ we have $n_\ell<r\le n$, so $N\setminus N_\ell\neq\emptyset$.)
We emphasise the two properties of $\varphi_\ell$ that the proof uses: it is injective, and its image avoids $N_\ell$.
No relation between $\varphi_\ell$ and $\varphi_{\ell'}$ is required for $\ell\neq\ell'$; a single agent may serve as a filler for many groups.

\item[Step 4 (the $r$ allocations).] Fix an arbitrary partition $(J_1,\ldots,J_n)$ of the junk goods $J$, the same in every round.
For every $t\in\mathbb{Z}_r$ and every agent $i\in N_\ell$, define
\[
    A^{(t)}_i
    := B^{\ell}_{\sigma_\ell(i)+t}
    \ \cup\!\!\!\bigcup_{\ell':\, i\in\varphi_{\ell'}(\mathbb{Z}_r\setminus S_{\ell'})}\!\!\!
      B^{\ell'}_{\varphi_{\ell'}^{-1}(i)+t}
    \ \cup\ J_i .
\]
\end{description}
The randomized allocation $\bm{A}$ is the \emph{uniform} distribution over $A^{(0)},\ldots,A^{(r-1)}$; its support has size $r\le n$.

\begin{lemma}[Validity]
\label{slem:laminar-validity}
For every $t\in\mathbb{Z}_r$, $A^{(t)}$ is an allocation of $M$.
\end{lemma}

\begin{proof}
Fix $t$ and $\ell$.
The translation $x\mapsto x+t$ is a bijection of $\mathbb{Z}_r$, so it maps the partition $\{S_\ell,\ \mathbb{Z}_r\setminus S_\ell\}$ of $\mathbb{Z}_r$ to the partition $\{S_\ell+t,\ (\mathbb{Z}_r\setminus S_\ell)+t\}$ of $\mathbb{Z}_r$.
In round $t$ the blocks of $\Gamma^{(\ell)}$ that are handed out are exactly those indexed by $S_\ell+t$---one to each member of $N_\ell$, distinct because $\sigma_\ell$ is injective---together with those indexed by $(\mathbb{Z}_r\setminus S_\ell)+t$---one to each filler, distinct because $\varphi_\ell$ is injective.
Hence every index $p\in\mathbb{Z}_r$ is issued exactly once, so the blocks of $\Gamma^{(\ell)}$ are distributed exactly once each.
As $\ell$ ranges over $[T]$ and the sets
$\Gamma^{(1)},\ldots,\Gamma^{(T)},J$ are pairwise disjoint with union $M$, and $(J_1,\ldots,J_n)$ partitions $J$, every good is assigned to exactly one agent.
Finally, the pieces constituting a single bundle $A^{(t)}_i$ are blocks of distinct approval sets together with $J_i$, hence pairwise disjoint.
\end{proof}

\begin{lemma}[Realized utilities and the cap]
\label{slem:laminar-cap}
Fix $t\in\mathbb{Z}_r$ and let $i\in N_\ell$.
Then
\[
    v_i\bigl(A^{(t)}_i\bigr)=\bigl|B^{\ell}_{\sigma_\ell(i)+t}\bigr|\in\{u,u+1\},
\]
and $v_i\bigl(A^{(t)}_j\bigr)\le u+1$ for every agent $j$.
\end{lemma}

\begin{proof}
Agent $i$ values only the goods of $\Gamma^{(\ell)}$, so blocks of other approval sets and junk goods contribute $0$ to any value computed by $i$; consequently $v_i(A^{(t)}_j)=|A^{(t)}_j\cap\Gamma^{(\ell)}|$ for every $j$.

We claim that every bundle contains at most one block of $\Gamma^{(\ell)}$.
Indeed, a member $j\in N_\ell$ receives her own block $B^{\ell}_{\sigma_\ell(j)+t}$ and no filler block of her own group, because the image of $\varphi_\ell$ avoids $N_\ell$; and an agent $j\notin N_\ell$ receives a block of $\Gamma^{(\ell)}$ only if $j=\varphi_\ell(q)$, in which case $q$ is unique by injectivity of $\varphi_\ell$, so she receives exactly the one block $B^{\ell}_{q+t}$.
Since every block has size at most $u+1$, this proves $v_i(A^{(t)}_j)\le u+1$.
Applying the claim to $j=i$ gives $v_i(A^{(t)}_i)=|B^{\ell}_{\sigma_\ell(i)+t}|$, which lies in $\{u,u+1\}$ by Step 1.
\end{proof}

\begin{lemma}[Ex-post guarantees]
\label{slem:laminar-expost}
Every $A^{(t)}$ is simultaneously \EFone{} and \EQone{}.
\end{lemma}

\begin{proof}
Fix $t$, write $A:=A^{(t)}$, and let $i\in N_\ell$ and $j$ be agents with $A_j\neq\emptyset$.

\emph{\EFone{}.}
If $v_i(A_j)\le v_i(A_i)$, then any $g\in A_j$ satisfies $v_i(A_i)\ge v_i(A_j\setminus\{g\})$ because values are nonnegative.
Otherwise \Cref{slem:laminar-cap} forces $v_i(A_j)=u+1$ and $v_i(A_i)=u$.
In particular $v_i(A_j)\ge1$, so $A_j$ contains a good $g$ with $v_i(g)=1$; for this $g$,
\[
    v_i(A_j\setminus\{g\})=u=v_i(A_i).
\]
Note that $g$ must be chosen among the goods that $i$ approves: deleting an arbitrary good of $A_j$---for instance a junk good, or a good of another group's approval set---need not decrease $v_i(A_j)$ at all.

\emph{\EQone{}.}
If $v_j(A_j)\le v_i(A_i)$, any $g\in A_j$ works, again by non-negativity.
(This case covers, in particular, a nonempty bundle $A_j$ with $v_j(A_j)=0$, which contains no good that $j$ approves; there the removed good is an arbitrary, zero-valued one.)
Otherwise \Cref{slem:laminar-cap}, applied to $j$ and to $i$, gives $v_j(A_j)=u+1$ and $v_i(A_i)=u$.
Then $A_j$ contains a good $g$ with $v_j(g)=1$, and
\[
    v_j(A_j\setminus\{g\})=u=v_i(A_i). \qedhere
\]
\end{proof}

\begin{lemma}[Ex-ante guarantees]
\label{slem:laminar-exante}
The randomized allocation $\bm{A}$ is ex-ante \EQ{} and ex-ante \EF{}; indeed every agent has expected utility exactly $k/r$.
\end{lemma}

\begin{proof}
Fix $i\in N_\ell$.
For fixed $\sigma_\ell(i)$, the map $t\mapsto\sigma_\ell(i)+t$ is a bijection of $\mathbb{Z}_r$, so by \Cref{slem:laminar-cap},
\[
    \sum_{t\in\mathbb{Z}_r} v_i\bigl(A^{(t)}_i\bigr)
    =\sum_{t\in\mathbb{Z}_r}\bigl|B^{\ell}_{\sigma_\ell(i)+t}\bigr|
    =\sum_{p\in\mathbb{Z}_r}\bigl|B^{\ell}_{p}\bigr|
    =\bigl|\Gamma^{(\ell)}\bigr|=k .
\]
As the randomized allocation is uniform over the $r$ rounds, $\mathbb{E}[v_i(A_i)]=k/r$.
Crucially, this value does not depend on $i$ or on $\ell$, because every approval set has the same size $k$; hence $\bm{A}$ is ex-ante \EQ{}.
(Here cutting every $\Gamma^{(\ell)}$ into the same number $r$ of blocks is used: cutting $\Gamma^{(\ell)}$ into $n_\ell$ blocks would give the members of a group with $n_\ell<r$ the larger expected utility $k/n_\ell$.)

For ex-ante \EF{}, fix $j\neq i$ and recall from the proof of \Cref{slem:laminar-cap} that $v_i(A^{(t)}_j)=|A^{(t)}_j\cap\Gamma^{(\ell)}|$ and that $A^{(t)}_j$ contains at most one block of $\Gamma^{(\ell)}$.
We distinguish three cases.
\begin{enumerate}[(a)]
    \item $j\in N_\ell$. Then $j$ receives the block
    $B^{\ell}_{\sigma_\ell(j)+t}$ in round $t$, and the computation above
    with $\sigma_\ell(j)$ in place of $\sigma_\ell(i)$ gives
    $\mathbb{E}[v_i(A_j)]=k/r$.
    \item $j=\varphi_\ell(q)$ for some (necessarily unique)
    $q\in\mathbb{Z}_r\setminus S_\ell$. Then $j$ receives the block
    $B^{\ell}_{q+t}$ in round $t$, and since $t\mapsto q+t$ is again a
    bijection of $\mathbb{Z}_r$,
    \[
        \mathbb{E}[v_i(A_j)]
        =\frac1r\sum_{t\in\mathbb{Z}_r}\bigl|B^{\ell}_{q+t}\bigr|
        =\frac{k}{r}.
    \]
    \item Otherwise $j$ receives no block of $\Gamma^{(\ell)}$ in any round,
    so $\mathbb{E}[v_i(A_j)]=0$.
\end{enumerate}
In every case $\mathbb{E}[v_i(A_j)]\le k/r=\mathbb{E}[v_i(A_i)]$, which is
ex-ante \EF{}.
\end{proof}

\begin{proof}[Proof of \Cref{thm:laminar-binary-lottery}]
By \Cref{slem:laminar-validity} the randomized allocation is supported on $r\le n$ deterministic allocations, by \Cref{slem:laminar-expost} every one of them is \EFone{}+\EQone{}, and by \Cref{slem:laminar-exante} the randomized allocation is ex-ante \EF{}+\EQ{}.
For the running time, the groups $N_\ell$ and the sets $\Gamma^{(\ell)}$ are obtained by comparing the $n$ approval sets pairwise, the balanced partitions of Step 1 and the injections of Steps 2 and 3 are constructed greedily, and the $r\le n$ allocations of Step 4 are then written down directly; every step is polynomial in $n$ and $m$.
\end{proof}

\begin{remark}
Ex-ante envy-freeness holds with equality rather than strictly: from agent $i$'s perspective, the $r$ agents in $N_\ell\cup\varphi_\ell(\mathbb{Z}_r\setminus S_\ell)$ all have expected value exactly $k/r$, and every other agent has expected value $0$.
Note also that our construction need not satisfy the ex-post guarantee \EQX{}: an agent with realized utility $u+1$ may also hold junk goods or blocks approved only by other groups, and removing such a zero-valued good from her bundle does not lower her utility, so an agent with realized utility $u$ still fails the \EQX{} comparison against her.
\end{remark}

\subsection{Tight Examples for Ex-post Guarantees in Randomized Allocations}
\label{sec:expost-limits}

Both \Cref{thm:laminar-binary-lottery} and \Cref{thm:chores-binary-ex-ante-ex-post} pair exact ex-ante \EF{}+\EQ{} with ex-post \EFone{}+\EQone{}.
In this section, we show that these ex-post guarantees cannot be strengthened to their ``up to any item'' counterparts.
For binary chores, the incompatibility arises from combining the ex-ante and ex-post requirements: \Cref{thm:binary_chores_efx_eqx} always guarantees a deterministic \EFX{}+\EQX{} allocation.

\begin{proposition}\label{prop:laminar-expost-efx-tight}
There is a normalized binary goods instance with four agents, five goods, and laminar approval sets in which no ex-ante \EQ{} randomized allocation is ex-post \EFX{}.
Consequently, the ex-post \EFone{} guarantee of \Cref{thm:laminar-binary-lottery} cannot be strengthened to \EFX{}.
\end{proposition}

\begin{proof}
Let $N=\{1,2,3,4\}$ and $M=\{a_1,a_2,b_1,b_2,z\}$, where agents $1,2,3$ approve $\Gamma_1=\Gamma_2=\Gamma_3=\{a_1,a_2\}$, agent~$4$ approves $\Gamma_4=\{b_1,b_2\}$, and the good $z$ is approved by nobody.
Every agent values $M$ at $2$, so the instance is normalized, and any two approval sets are either identical or disjoint, and hence the family is laminar.
We claim that every \EFX{} allocation $A$ satisfies
\begin{equation}\label{eq:laminar-tight}
v_4(A_4)\ \ge\ 1
\qquad\text{and}\qquad
\sum_{i=1}^{3} v_i(A_i)\ =\ 2 .
\end{equation}

\emph{The first part of \eqref{eq:laminar-tight}.}
Suppose $v_4(A_4)=0$, i.e., $b_1,b_2\notin A_4$.
If a single agent $i$ held both $b_1$ and $b_2$, then $i\neq4$ and, taking $g=b_1$, we would get $v_4(A_i\setminus\{g\})\ge v_4(b_2)=1>0=v_4(A_4)$, contradicting \EFX{} for the pair $(4,i)$.
Hence $b_1\in A_{i_1}$ and $b_2\in A_{i_2}$ for two distinct agents $i_1,i_2\in\{1,2,3\}$.
Moreover, if $A_{i_1}$ contained a good $g\neq b_1$, then $v_4(A_{i_1}\setminus\{g\})\ge v_4(b_1)=1>0$, again contradicting \EFX{} for $(4,i_1)$.
Therefore $A_{i_1}=\{b_1\}$ and, symmetrically, $A_{i_2}=\{b_2\}$; in particular $v_{i_1}(A_{i_1})=0$.
The remaining goods $a_1,a_2,z$ are thus split between agent~$4$ and the unique remaining agent $y\in\{1,2,3\}$.
For $x\in\{4,y\}$, \EFX{} for the pair $(i_1,x)$ requires $v_{i_1}(A_x\setminus\{g\})\le v_{i_1}(A_{i_1})=0$ for every $g\in A_x$; consequently, if $A_x$ contains one of $a_1,a_2$, then $A_x$ contains nothing else.
Consequently, if either $A_4$ or $A_y$ contains one of $a_1, a_2$, that bundle must be a singleton. But $A_4\cup A_y$ must contain all three goods $a_1, a_2, z$.
Since these three goods must be distributed between only two bundles, some bundle must contain one of $a_1,a_2$ together with another good, a contradiction.

\emph{The second part of \eqref{eq:laminar-tight}.}
No agent holds both $a_1$ and $a_2$: if agent $x$ did, pick any $i\in\{1,2,3\}\setminus\{x\}$, which exists because $|\{1,2,3\}|=3$; then $v_i(A_i)=0$ while $v_i(A_x\setminus\{a_1\})\ge v_i(a_2)=1$, contradicting \EFX{} for $(i,x)$.
Next, agent~$4$ holds neither $a_1$ nor $a_2$: suppose $a_1\in A_4$.
By the first part, $A_4$ also contains a good of $\{b_1,b_2\}$, so $|A_4|\ge2$; and since $a_2\notin A_4$, at least two agents of $\{1,2,3\}$ receive no good of $\{a_1,a_2\}$ and hence have utility $0$.
Picking such an agent $i$ and any $g\in A_4\setminus\{a_1\}$ gives $v_i(A_4\setminus\{g\})\ge v_i(a_1)=1>0=v_i(A_i)$, contradicting \EFX{} for $(i,4)$.
Hence $a_1$ and $a_2$ are held by two distinct agents of $\{1,2,3\}$, so exactly two of these agents have utility $1$ and the third has utility $0$, which gives $\sum_{i\le3}v_i(A_i)=2$.

Now let $\bm{A}$ be any ex-post \EFX{} randomized allocation and write $\mu_i=\mathbb{E}[v_i(A_i)]$.
Taking expectations in \eqref{eq:laminar-tight} yields $\mu_1+\mu_2+\mu_3=2$ and $\mu_4\ge1$.
If $\bm{A}$ were ex-ante \EQ{}, all four expectations would equal a common value $\mu$, so that $3\mu=2$ and simultaneously $\mu\ge1$, which is impossible.
\end{proof}

\begin{proposition}\label{prop:chores-expost-tight}
There is a binary chores instance with two agents and two chores in which no ex-ante \EF{}+\EQ{} randomized allocation is ex-post \EFX{}, and none is ex-post \EQX{}.
Consequently, neither the ex-post \EFone{} nor the ex-post \EQone{} guarantee of \Cref{thm:chores-binary-ex-ante-ex-post} can be strengthened to \EFX{} or \EQX{}, respectively.
\end{proposition}

\begin{proof}
Let $N=\{1,2\}$ and $M=\{c_1,c_2\}$ with $v_1(c_1)=v_2(c_1)=v_2(c_2)=-1$ and $v_1(c_2)=0$.
The instance is binary, hence restricted additive with base value $v(c)=-1$, and it is unnormalized.
If agent~$1$ receives both chores, then both \EFX{} and \EQX{} fail for the pair $(1,2)$ with $c=c_2$, because $v_1(A_1\setminus\{c_2\})=-1<0=v_1(A_2)=v_2(A_2)$; symmetrically, both fail with $c=c_1$ when agent~$2$ receives both chores.
The two remaining allocations,
\[
A^{\ast}=(\{c_1\},\{c_2\})
\qquad\text{and}\qquad
A^{\dagger}=(\{c_2\},\{c_1\}),
\]
are \EFX{} and \EQX{}: every bundle is a singleton, so deleting its unique chore leaves the empty bundle, of value $0$, which is at least the value of any bundle in a chores instance.
Hence every ex-post \EFX{} (or ex-post \EQX{}) randomized allocation is supported on $\{A^{\ast},A^{\dagger}\}$.
Writing $p$ for the probability of $A^{\ast}$, we obtain $\mathbb{E}[v_1(A_1)]=-p$ and $\mathbb{E}[v_2(A_2)]=-1$, so ex-ante \EQ{} forces $p=1$, i.e., the deterministic allocation $A^{\ast}$.
But $A^{\ast}$ is not ex-ante \EF{}, since $v_1(A_1)=-1<0=v_1(A_2)$.
\end{proof}

\begin{remark}\label{rem:laminar-eqx-open}
The six-agent instance in \Cref{sprop:six-agent-ef1-eqx} has laminar approval sets and admits no deterministic \EFone{}+\EQX{} allocation. Thus \Cref{thm:laminar-binary-lottery} cannot be strengthened to ex-post \EFone{}+\EQX{}, even without any ex-ante requirements.
\end{remark}

\end{document}